\documentclass{scrartcl}
\usepackage[numbers]{natbib}
\usepackage{fontenc}
\usepackage{shadethm}
\usepackage{amsthm}
\usepackage{float}
\usepackage{bm}
\usepackage{mathtools}
\usepackage{graphicx}
\usepackage{subfig}
\usepackage{stmaryrd}
\usepackage{mathrsfs} 
\usepackage{amsmath}
\usepackage{amssymb}
\usepackage{wasysym}
\usepackage{sectsty}
\usepackage[hyperfootnotes=false]{hyperref}
\usepackage[a4paper, left=2.7cm, right=2.7cm, top=2.5cm]{geometry}
\usepackage{chngcntr}
\counterwithin*{equation}{section}

\usepackage{cleveref}
\DeclareMathAlphabet{\mathpzc}{OT1}{pzc}{m}{it}

\sectionfont{\normalsize\centering}
\subsectionfont{\footnotesize\centering}

\theoremstyle{plain}
\newtheorem{thm}{Theorem}[section] 

\theoremstyle{definition}
\newtheorem{defn}[thm]{Definition} 
\newtheorem{lem}[thm]{Lemma}

\newtheorem{rem}[thm]{Remark}

\newcommand*{\Scale}[2][4]{\scalebox{#1}{$#2$}}%

\def\XXint#1#2#3{{\setbox0=\hbox{$#1{#2#3}{\int}$ }
		\vcenter{\hbox{$#2#3$ }}\kern-.6\wd0}}

\usepackage[utf8]{inputenc}
\usepackage[german,english,russian]{babel}

\newcounter{MPequ}

\newcounter{AppA}

\newcounter{AppB}

\newcounter{AppC}

\newcounter{AppD}

\newcounter{AppE}

\begin{document}\selectlanguage{english}
\begin{center}
\normalsize \textbf{\textsf{Charged particle motion in a strong magnetic field: The first order expansion}}
\end{center}
\begin{center}
	Ugo Boscain$^{\star,}$\footnote{\textit{E-mail address:} \href{mailto:Ugo.boscain@inria.fr}{ugo.boscain@inria.fr}}, Wadim Gerner$^{\dagger,}$\footnote{\textit{E-mail address:} \href{mailto:wadim.gerner@edu.unige.it}{wadim.gerner@edu.unige.it}}
\end{center}
\begin{center}
{\footnotesize Laboratoire Jacques-Louis Lions, Sorbonne Universit\'{e}, Universit\'{e} de Paris, CNRS, Inria, Paris, France$^{\star}$}
\newline
\newline
{\footnotesize MaLGa Center, Department of Mathematics, Department of Excellence 2023-2027, University of Genoa, Via Dodecaneso 35, 16146 Genova, Italy$^{\dagger}$}
\end{center}
{\small \textbf{Abstract:} 
We provide a mathematically rigorous derivation of the first order expansion of the motion of a charged particle in a strong magnetic field. In contrast to the derivations that can be found in the physics literature we solely assume throughout that the magnetic field is strong. In particular we do not need to make any structural assumptions on the particle motion, such as the gyroradius being small in comparison to the magnetic length scale. Instead, some of the additional assumptions which are usually made in the physics literature turn out to be an a posteriori consequence in our approach. Our approach further justifies the utilisation of the guiding centre approximation at "bounce points" within magnetic mirrors, a situation which violates the usual assumptions which are made in the physics literature when deriving the guiding centre approximation.

\section{Introduction}
\label{Introduction}

Consider a particle of mass  $m$, charge  $q$ and position $x$ moving in $\mathbb{R}^3$ under the action of a  (static) magnetic field $\bm{B}$. Its motion is described by the equation
\begin{gather}
	\label{S1E1}
	m\ddot{x}=q\dot{x}\times \bm{B}(x).
\end{gather}

We introduce a reference magnetic field strength $B_{\operatorname{ref}}>0$ (which we for instance may take to be the average magnetic field strength) and define the \textbf{reference gyrofrequency} $\omega:=\frac{q B_{\operatorname{ref}}}{m}$ together with the normalised magnetic field $B:=\frac{\bold{B}}{B_{\operatorname{ref}}}$. Then (\ref{S1E1}) becomes
\begin{gather}
	\label{S1E2}
	\ddot{x}_{\omega}=\omega \dot{x}_{\omega}\times B(x_{\omega}).
\end{gather}
where the subscript $\omega$ indicates the dependence of the particle trajectory on the value of $\omega$.

Although \eqref{S1E2} is easy to write, its explicit solution is hard to find in general. In many applications (as for instance in magnetic confinement for plasma fusion) it is crucial to understand the asymptotic behaviour of $x_\omega$ in the regime of large $\omega$. In particular, under suitable regularity assumptions on the magnetic field $B$ (for instance $B\in C^2_{\operatorname{loc}}(\mathbb{R}^3,\mathbb{R}^3)$), one would like to prove an expansion of the form
\begin{gather}
\label{espansione}
x_\omega(t)={\mathbf{x}}_0(t)+\frac1\omega {\mathbf{x}}_1(t)+o\left(\frac1\omega\right) \text{ as }\omega\rightarrow\infty
\end{gather}
and to compute an explicit expression of ${\mathbf{x}}_0(t)$ and ${\mathbf{x}}_1(t)$. 
Here  $o(\frac1\omega)$ is understood in a suitable functional space, for instance $C^0_{\operatorname{loc}}(\mathbb{R})$.

In the physics literature this expansion has been deeply studied starting in the 1940s with the pioneering work of Alfvén
\cite{alfven1940motion}, followed by more systematic contributions of Northrop where the different drift terms were identified \cite{northrop1961guiding,northrop1963adiabatic}, and culminating in the modern variational approach developed by Littlejohn \cite{littlejohn1983variational}. 
See also the more recent references \cite{Chen16,IGPW20,IGPW24}.

Proving rigorously the validity of an expansion of the type \eqref{espansione} and obtaining explicit expressions for
${\mathbf{x}}_0(t)$ and ${\mathbf{x}}_1(t)$  is however a delicate task. In the physics literature (see the references cited above), this program is typically carried out by introducing additional assumptions on the behaviour of the solution $x_\omega$ for $\omega$ large. We shall discuss this point in detail in Section~\ref{s-physics}. 

	Typically, in the physics literature, some additional a priori structural assumptions are being made on the solution $x_{\omega}$ as well as its derivatives. The following two assumptions are archetypical, see Section~\ref{s-physics} for a more thorough discussion:
\begin{itemize}
	\item It is assumed that the particle motion $x_{\omega}$ may be decomposed into a main "guiding centre" motion $R_{\omega}$ and an oscillating "gyromotion" $\rho_{\omega}$. In addition, it is further assumed that the gyromotion is small in comparison to the guiding centre motion or in other words that the guiding centre is "dominant".
	\item It is further assumed, based on statistical arguments, that the particle velocity which is parallel to the magnetic field is of the same order of magnitude as the thermal velocity of the plasma. This assumption is used during the course of the derivation in the physics literature in order to argue that certain terms become negligible. However, there are relevant physical situations in which this assumption is no longer valid, such as in magnetic mirrors.
\end{itemize} 
In this sense, the existing derivations rely on a bootstrap argument in which certain properties of $x_\omega$ are assumed a priori, rather than being derived from the dynamics.

In other words, while the expansion \eqref{espansione} has proved extremely effective in applications, 
its justification is still incomplete from a mathematical perspective. 
A rigorous derivation should establish \eqref{espansione} under assumptions on the magnetic field $B$ alone, 
without presupposing structural assumptions of the trajectory $x_\omega$.

\bigskip

The zero order expansion, also known as the zeroth order guiding centre motion in the physics literature, has been investigated recently in \cite{BG25} in a mathematically rigorous manner. The main result in this regards (which contains in particular ${\mathbf{x}}_0(t)$) is the following

\begin{thm}[{\cite[Theorem 1.2 \& Corollary 1.4]{BG25}}]
		\label{S2T1}
		Let $B\in C^2_{\operatorname{loc}}(\mathbb{R}^3,\mathbb{R}^3)$ be a no-where vanishing div-free vector field and let $x_0,v_0\in \mathbb{R}^3$. Let further $x_{\omega}$ denote the (unique) solution to (\ref{S1E2}) with initial conditions $x_0,v_0$. Then there exists some $x\in C^3_{\operatorname{loc}}(\mathbb{R},\mathbb{R}^3)$ such that
		\begin{gather}
			\nonumber
			x_{\omega}\rightarrow x\text{ in }C^{0,\alpha}_{\operatorname{loc}}(\mathbb{R},\mathbb{R}^3)\text{ for all }0\leq \alpha<1\text{ as }\omega\rightarrow\infty.
		\end{gather}
		In addition, $\dot{x}(t)= h(t)b(x(t))$ for a suitable $h\in C^2_{\operatorname{loc}}(\mathbb{R})$ function where $b(x):=\frac{B(x)}{|B(x)|}$. Further, the magnetic moment $\mu(t):=\frac{|v_0|^2-h^2(t)}{2|B(x(t))|}$ is time-independent.
	\end{thm}
	We point out that in the original reference \cite[Theorem 1.2]{BG25} it is assumed that the magnetic field $B$ is of class $C^2$ and that its derivatives up to second order are globally bounded and that $B$ has a certain decay at infinity. This is necessary in order to derive a semi-explicit convergence rate of $x_{\omega}$ to $x$. If $B$ is solely of class $C^2$, then we still obtain a $C^{0,\alpha}_{\operatorname{loc}}$ convergence result with an unknown convergence rate. Since we will not derive any convergence rate in the present manuscript, we always assume throughout that $B$ is of class $C^2$, but with possibly (globally) unbounded derivatives.
	\newline
	\newline
	The main goal of this manuscript is to obtain the first order correction term ${\mathbf{x}}_1(t)$ in the expansion of $x_{\omega}$ which describes the drift of the particle normal to the direction of the magnetic field.
	
	Our first main result is the following
	\begin{thm}[First order expansion (full particle trajectory)]
		\label{S2T2}
		Let $B\in C^2_{\operatorname{loc}}(\mathbb{R}^3,\mathbb{R}^3)$ be a no-where vanishing div-free vector field, $x_0,v_0\in \mathbb{R}^3$ and $x_{\omega}$ denote the solution to (\ref{S1E2}) with initial conditions $x_0,v_0$. Let further $x(t),h(t),\mu(t)$ be as in \Cref{S2T1} and set $b(y):=\frac{B(y)}{|B(y)|}$, $\mu_0:=\mu(0)$. Then the following holds
		\begin{gather}
			\nonumber
			x_{\omega}(t)=x_0+\int_0^t\left(h(s)+o\left(1\right)\right)b(x_{\omega}(s))ds+\frac{b(x(t))\times v_{\omega}(t)}{\omega|B(x(t))|}
			\\
			\label{S2E1}
			+\frac{1}{\omega}\left(\frac{v_0\times b(x_0)}{|B(x_0)|}+\int_0^t\frac{b(x_{\omega}(s))}{|B(x(s))|}\times \left(\mu_0\nabla (|B|)(x(s))+h^2(s)\kappa(x(s))+o\left(1\right)\right)ds\right)\text{ in }C^0_{\operatorname{loc}}(\mathbb{R})\text{ as }\omega\rightarrow\infty
		\end{gather}
		where $\kappa(y):=(\nabla_bb)(y)$ is the magnetic curvature and $v_{\omega}(t)=\dot{x}_{\omega}(t)$.
	\end{thm}
	\begin{rem}
		We point out that as commonly done in the physics literature we decompose the motion of $x_{\omega}$ into a motion parallel to the magnetic field given by $\int_0^t(h(s)+o(1))b(x_{\omega}(s))ds$ and a motion perpendicular to the magnetic field given by the second line in (\ref{S2E1}). As we shall see below, the remaining term containing $v_{\omega}$ essentially encompasses an oscillation. We notice that the velocity of the main motion in parallel direction of the magnetic field is given by $(h(s)+o(1))b(x_{\omega}(s))$ and hence is dominated by the term $h(s)$ which is of order $\mathcal{O}(1)$. In contrast, the order of the motion in perpendicular direction is of order $\mathcal{O}\left(\frac{1}{\omega}\right)$. The goal of \Cref{S2T2} is to provide the leading order terms for the parallel as well as perpendicular motion. This approach is customary in the physics literature. We point out that with some additional work one should be able to recover the second order correction term in parallel direction which should be of order $\mathcal{O}\left(\frac{1}{\omega}\right)$. We, however, do not pursue this in the present manuscript. Lastly, let us point out that if we denote by $x(t)$ the zero order guiding centre approximation as in \Cref{S2T1}, then we have $x(t)=x_0+\int_0^t(h(s)+o(1))b(x_{\omega}(s))ds+o(1)$ in $C^0_{\operatorname{loc}}(\mathbb{R})$ as $\omega\rightarrow\infty$, so that the zero order approximation is implicitly contained in the expansion (\ref{S2E1}).
	\end{rem}
	
	We observe that (\ref{S1E2}) implies that $|v_{\omega}(t)|=|v_0|$ for all $t$ and $\omega$. Hence, we find $\frac{b(x(t))\times v_{\omega}(t)}{\omega |B(x(t))|}\in \mathcal{O}\left(\frac{1}{\omega}\right)$. Even more, upon time-averaging this term will not contribute significantly to the particle trajectory in the following sense:
	
	For any fixed $t_0\geq 0$ we have
	\begin{gather}
		\nonumber
		\frac{\int_{t_0}^t\frac{b(x(s))\times v_{\omega}(s)}{\omega |B(x(s))|}ds}{t-t_0}\in o\left(\frac{1}{\omega}\right) \text{ in } C^0_{\operatorname{loc}}((t_0,\infty)) \text{ as } \omega\rightarrow\infty
	\end{gather}
	which follows immediately from \Cref{S2T1},(\ref{S1E2}) and an integration by parts.
	\newline
	\newline
	We notice that at this point we solely considered the full particle trajectory $x_{\omega}(t)$ and have not performed any decomposition into a guiding centre motion and a gyromotion. To this end we make the following definitions.
	\begin{defn}[Guiding centre and gyromotion]
		\label{S2D3}
		Let $B\in C^2_{\operatorname{loc}}(\mathbb{R}^3,\mathbb{R}^3)$ be a no-where vanishing div-free vector field and $x_0,v_0\in \mathbb{R}^3$. With the notation from \Cref{S2T2} we define the \textit{gyromotion} as follows
		\begin{gather}
			\label{S2E2}
			\rho_{\omega}(t):=\frac{b(x_{\omega}(t))\times v_{\omega}(t)}{\omega|B(x_{\omega}(t))|}
		\end{gather}
		and we define the \textit{guiding centre} as follows
		\begin{gather}
			\label{S2E3}
			R_{\omega}(t):=x_{\omega}(t)-\rho_{\omega}(t).
		\end{gather}
	\end{defn}
	With these definition we obtain the following conclusion from \Cref{S2T2}.
	\begin{thm}
		\label{S2T4}
		Let $B\in C^2_{\operatorname{loc}}(\mathbb{R}^3,\mathbb{R}^3)$ be a div-free no-where vanishing vector field and $x_0,v_0\in \mathbb{R}^3$. Using the notation from \Cref{S2T2} we have
		\begin{gather}
			\nonumber
			R_{\omega}(t)=x_0+\int_0^t(h(s)+o(1))b(R_{\omega}(s))ds
			\\
			\label{S2E4}
			+\frac{1}{\omega}\left(\frac{v_0\times b(x_0)}{|B(x_0)|}+\int_0^t\frac{b(R_{\omega}(s))}{|B(R_{\omega}(s))|}\times \left(\mu_0\nabla (|B|)(R_{\omega}(s))+h^2(s)\kappa(R_{\omega}(s))+o(1)\right)ds\right)
		\end{gather}
		in $C^0_{\operatorname{loc}}(\mathbb{R})$ as $\omega\rightarrow\infty$.
	\end{thm}
We notice that our expansion (\ref{S2E4}) contains a "curvature-drift" which is parallel to $b\times \kappa$ and caused by the curvature $\kappa$ of the magnetic field lines and a "grad-$B$-drift" which is parallel to $b\times \nabla |B|$ and caused by a change in magnetic field strength. Both of these terms also appear in the physics literature, cf. \cite[Chapter 4.3.6]{IGPW24}. A more detailed comparison of our results with the results in the physics literature may be found in \Cref{s-physics-math}. In particular, we include a detailed discussion about the necessity of the assumptions in order for the guiding centre approximation to hold.

The structure of the paper is the following. In Section \ref{s-math} we prove the main result of the paper (namely Theorem \ref{S2T4}).
The main idea is to express $x_{\omega}(t)-x_0=\int_0^t\dot{x}_{\omega}(s)ds$ and to integrate by parts while repeatedly exploiting (\ref{S1E2}). To control the appearing terms and identify the corresponding limits one needs to use non-standard calculus identities, cf. \Cref{5L4}.
In Section \ref{s-physics} we discuss the approach in the physics literature following mainly \cite{IGPW20,IGPW24}.
In Section~\ref{s-physics-math} we compare our assumptions and results with those commonly adopted in the physics literature, and show their mutual consistency. 
Finally, in Section~\ref{s-final} we present concluding remarks, including a discussion of 
{\textbf{i)}} the fact that gyromotion always occurs in the direction opposite to the curvature of the charged particle trajectory, and 
{\textbf{ii)}} the case of plasma equilibria, where we present a different way to express the particle drift. This equivalent expression captures more naturally the drift away from magnetic surfaces and is therefore better suited for plasma confinement purposes than the expression commonly found in the physics literature.

\section{Proof of mathematical results}
\label{s-math}
\subsection{A useful identity}
We start by proving a useful calculus formula.
\begin{lem}
	\label{5L4}
	Let $v\in \mathbb{R}^3$ and $b:\mathbb{R}^3\rightarrow\mathbb{R}^3$ be a $C^1$-unit vector field, i.e. $|b(x)|=1$ for all $x\in \mathbb{R}^3$. Then for every $x\in \mathbb{R}^3$
	\begin{gather}
		\nonumber
		(Db(x)\cdot (v\times b(x)))\times (v\times b(x))
		\\
		\nonumber
		=\left((b\cdot v)\cdot (v\cdot \operatorname{curl}(b)(x))-|v|^2(b(x)\cdot \operatorname{curl}(b)(x))-\nabla_vb(x)\cdot (v\times b(x))\right)b(x)
	\end{gather}
\end{lem}
\begin{proof}[Proof of \Cref{5L4}]
	We start by exploiting the identity
	\begin{gather}
		\nonumber
		0=\nabla (b\cdot (v\times b))=\nabla_{v\times b}b+\nabla_b(v\times b)+(v\times b)\times \operatorname{curl}(b)+b\times (\operatorname{curl}(v\times b))
		\\
		=\nabla_{v\times b}b+v\times \nabla_bb+(v\times b)\times \operatorname{curl}(b)+b\times \operatorname{curl}(v\times b)
	\end{gather}
	where we used that $\nabla_b(v\times b)=\nabla_bv\times b+v\times \nabla_bb$ and $\nabla_bv=0$ since $v$ is a fixed vector independent of $x$. We now use the identity $\operatorname{curl}(v\times b)=-[v,b]+v\operatorname{div}(b)-b\operatorname{div}(v)=-\nabla_vb+v\operatorname{div}(b)$ where we used again that $v$ is independent of $x$. We conclude
	\begin{gather}
		\nonumber
		\nabla_{v\times b}b=\nabla_bb\times v+\operatorname{curl}(b)\times (v\times b)-\nabla_vb\times b+\operatorname{div}(b)v\times b.
	\end{gather}
	Now we note that $\nabla_bb=\operatorname{curl}(b)\times b$ since $b$ is a unit field and so we find $\nabla_bb\times v+\operatorname{curl}(b)\times (v\times b)=(\operatorname{curl}(b)\times v)\times b$ according to the Jacobi identity. We obtain
	\begin{gather}
		\label{5E25}
		\nabla_{v\times b}b=(\operatorname{curl}(b)\times v)\times b-\nabla_vb\times b+\operatorname{div}(b)v\times b.
	\end{gather}
	Now we observe that $Db\cdot (v\times b)=\nabla_{v\times b}b$ so that
	\begin{gather}
		\label{5E26}
		(Db\cdot (v\times b))\times (v\times b)=((\operatorname{curl}(b)\times v)\times b)\times (v\times b)-(\nabla_vb\times b)\times (v\times b).
	\end{gather}
	We compute $(\operatorname{curl}(b)\times v)\times b=(b\cdot \operatorname{curl}(b))v-(b\cdot v)\operatorname{curl}(b)$ and so $((\operatorname{curl}(b)\times v)\times b)\times (v\times b)=(v\cdot b)(v\cdot \operatorname{curl}(b))b-|v|^2(b\cdot \operatorname{curl}(b))b$. On the other hand $(\nabla_vb\times b)\times (v\times b)=(\nabla_vb\cdot (v\times b))b$ since $b$ and $v\times b$ are orthogonal. Therefore (\ref{5E26}) becomes
	\begin{gather}
		\nonumber
		(Db\cdot (v\times b))\times (v\times b)=(b\cdot v)(v\cdot \operatorname{curl}(b))b-|v|^2(b\cdot \operatorname{curl}(b))b-(\nabla_vb\cdot (v\times b))b
	\end{gather}
	which is the claimed identity.
\end{proof}
\subsection{Full particle trajectory}
\begin{proof}[Proof of \Cref{S2T2}]
	We start by expressing the full particle trajectory $x_{\omega}(t)$ in the following way
	\begin{gather}
		\label{S5ExtraExtra1}
		x_{\omega}(t)=x_0+\int_0^tv_{\omega}(s)ds=x_0+\int_0^tv_{\omega}^\parallel(s)ds+\int_0^tv^\perp_{\omega}(s)ds
	\end{gather}
	where $v^\parallel_{\omega}(s):=(v_{\omega}(s)\cdot b(x_{\omega}(s)))b(x_{\omega}(s))$ and $v^\perp_{\omega}(s):=v_{\omega}(s)-v^\parallel_{\omega}(s)$. It follows from the proof of \Cref{S2T1} in \cite{BG25} that $v_{\omega}(s)\cdot b(x_{\omega}(s))\rightarrow h(s)$ in $C^{0,\alpha}_{\operatorname{loc}}$ for every $0\leq \alpha<1$ as $\omega\rightarrow\infty$, where $h$ is defined as in \Cref{S2T1}. Consequently we arrive at the identity
	\begin{gather}
		\label{S5Extra1}
		x_{\omega}(t)=x_0+\int_0^t(h(s)+o(1))b(x_{\omega}(s))ds+\int_0^tv^\perp_{\omega}(s)ds\text{ in }C^{0,\alpha}_{\operatorname{loc}}\text{ for all }0\leq \alpha<1\text{ as }\omega\rightarrow\infty.
	\end{gather} 
	For convenience we define
	\begin{gather}
		\label{S5Extra2}
		z_{\omega}(t):=\omega\int_0^tv^\perp_{\omega}(s)ds.
	\end{gather}
	We introduce the notation 
	\begin{gather}
		\nonumber
		b_{\omega}(s):=b(x_{\omega}(s))\text{, }B_{\omega}(s):=B(x_{\omega}(s))
	\end{gather}
	and notice first that we can express
	\begin{gather}
		\label{5E5}
		v^{\perp}_{\omega}=b_{\omega}\times (v_{\omega}\times b_{\omega}).
	\end{gather}
	Further we recall, see (\ref{S1E2}), that we have the identity
	\begin{gather}
		\label{5E6}
		\dot{v}_{\omega}=\omega (v_{\omega}\times B_{\omega}).
	\end{gather}
	We hence find
	\begin{gather}
		\nonumber
		v^\perp_{\omega}=\frac{b_{\omega}}{|B_{\omega}|\omega}\times \dot{v}_{\omega}.
	\end{gather}
	Consequently we find
	\begin{gather}
		\nonumber
		z_{\omega}(t)=\int_0^t \frac{b_{\omega}(s)}{|B_{\omega}(s)|}\times \dot{v}_{\omega}(s)ds.
	\end{gather}
	We can integrate by parts and find
	\begin{gather}
		\label{5E7}
		z_{\omega}(t)=\frac{b_{\omega}(s)}{|B_{\omega}(s)|}\times v_{\omega}(s)\Bigg|^t_0-\int_0^t\frac{d}{ds}\left(\frac{b_{\omega}(s)}{|B_{\omega}(s)|}\right)\times v_{\omega}(s)ds.
	\end{gather}
	We now compute, using the notation $Db_{\omega}(s):=Db(x_{\omega}(s))$ and similar for other derivatives,
	\begin{gather}
		\nonumber
		\frac{d}{ds}\left(\frac{b_{\omega}}{|B_{\omega}|}\right)=\frac{Db_{\omega}(s)\cdot v_{\omega}(s)}{|B_{\omega}(s)|}+b_{\omega}(s)\frac{d}{ds}(|B_{\omega}(s)|^{-1})
		\\
		\label{5E8}
		=\frac{Db_{\omega}(s)\cdot v_{\omega}(s)}{|B_{\omega}(s)|}-\frac{\frac{\nabla |B_{\omega}(s)|^2}{2}\cdot v_{\omega}(s)}{|B_{\omega}(s)|^4}B_{\omega}(s).
	\end{gather}
	Inserting (\ref{5E8}) into (\ref{5E7}) yields
	\begin{gather}
		\nonumber
		z_{\omega}(t)=\frac{b_{\omega}(s)}{|B_{\omega}(s)|}\times v_{\omega}(s)\Bigg|^t_0-\int_0^t\frac{Db_{\omega}(s)\cdot v_{\omega}(s)}{|B_{\omega}(s)|}\times v_{\omega}(s)ds
		\\
		\label{5E9}
		+\int_0^t\frac{\frac{\nabla |B_{\omega}(s)|^2}{2}\cdot v_{\omega}(s)}{|B_{\omega}(s)|^4}\left(B_{\omega}(s)\times v_{\omega}(s)\right)ds.
	\end{gather}
	We now first deal with the term $\frac{b_{\omega}(t)}{|B_{\omega}(t)|}\times v_{\omega}(t)$. We note that according to \Cref{S2T1} $x_{\omega}(t)\rightarrow x(t)$ in $C^{0,\alpha}_{\operatorname{loc}}$ and that $|v_{\omega}(t)|=|v_0|$ for all $t$ and $\omega$ according to the fact that $\dot{v}_{\omega}$ is orthogonal to $v_{\omega}$, cf. (\ref{S1E2}). Hence
	\begin{gather}
		\nonumber
		\frac{b_{\omega}(t)}{|B_{\omega}(t)|}\times v_{\omega}(t)=\frac{b(x(t))\times v_{\omega}(t)}{|B(x(t))|}+o(1)\text{ in }C^{0}_{\operatorname{loc}}\text{ as }\omega\rightarrow\infty.
	\end{gather}
	We obtain
	\begin{gather}
		\nonumber
		z_{\omega}(t)=\frac{b(x(t))\times v_{\omega}(t)}{|B(x(t))|}-\frac{b(x_0)\times v_0}{|B(x_0)|}-\int_0^t\frac{Db_{\omega}(s)\cdot v_{\omega}(s)}{|B_{\omega}(s)|}\times v_{\omega}(s)ds
		\\
		\label{S5Extra3}
		+\int_0^t\frac{\frac{\nabla |B_{\omega}(s)|^2}{2}\cdot v_{\omega}(s)}{|B_{\omega}(s)|^4}\left(B_{\omega}(s)\times v_{\omega}(s)\right)ds+o(1)\text{ in }C^0_{\operatorname{loc}}\text{ as }\omega\rightarrow\infty.
	\end{gather}
	We observe further that $v^\parallel_{\omega}(s)=h_{\omega}b_{\omega}$ with $h_{\omega}(s):=(b_{\omega}\cdot v_{\omega})$. It then follows as in the beginning of this proof that $h_{\omega}\rightarrow h$ in $C^0_{\operatorname{loc}}$ and similarly that $b(x_{\omega}(s))\rightarrow b(x(s))$ in $C^0_{\operatorname{loc}}$ where $x(s)$ and $h(s)$ are as in \Cref{S2T1}. With this in mind we decompose the first integrand in (\ref{S5Extra3}) as follows
	\begin{gather}
		\nonumber
		\int_0^t\frac{Db_{\omega}\cdot v_{\omega}}{|B_{\omega}|}\times v_{\omega}ds=\int_0^t\frac{Db_{\omega}\cdot v^\parallel_{\omega}}{|B_{\omega}|}\times v^\parallel_{\omega}ds+\int_0^t\frac{Db_{\omega}\cdot v^\parallel_{\omega}}{|B_{\omega}|}\times v^\perp_{\omega}ds
		\\
		\label{5E11}
		+\int_0^t\frac{Db_{\omega}\cdot v^{\perp}_{\omega}}{|B_{\omega}|}\times v^\parallel_{\omega}ds+\int_0^t\frac{Db_{\omega}\cdot v^\perp_{\omega}}{|B_{\omega}|}\times v^\perp_{\omega}ds.
	\end{gather}
	The first term on the right hand side of (\ref{5E11}) converges then locally uniformly to
	\begin{gather}
		\label{5E12}
		\int_0^t\frac{Db_{\omega}\cdot v^\parallel_{\omega}}{|B_{\omega}|}\times v^\parallel_{\omega}ds=\int_0^t\frac{Db(x(s))\cdot v(s)}{|B(x(s))|}\times v(s)ds+o(1)\text{ in }C^0_{\operatorname{loc}}\text{ as }\omega\rightarrow\infty
	\end{gather}
	where $v(s)=\dot{x}(s)$, recall \Cref{S2T1}. As for the mixed terms, we observe that for instance $\frac{Db_{\omega}\cdot v^\parallel_{\omega}}{|B_{\omega}|}$ converges locally uniformly to $\frac{Db(x)\cdot v}{|B(x)|}$. On the other hand it is straightforward to verify by means of (\ref{5E5}) and (\ref{5E6}) that $v^\perp_{\omega}$ converges to zero in the distributional sense. Further, $|v^\perp_{\omega}(t)|\leq |v_{\omega}(t)|=|v_0|$ is bounded and thus locally bounded in every $L^p$-norm, $1\leq p\leq \infty$. From here we conclude that $v^\perp_{\omega}$ converges weakly in $L^p_{\operatorname{loc}}$ to $0$ for every $1\leq p<\infty$. Consequently $\frac{Db_{\omega}\cdot v^\parallel_{\omega}}{|B_{\omega}|}\times v^\perp_{\omega}$ converges weakly in $L^2_{\operatorname{loc}}$ to $0$. But since $\frac{Db_{\omega}\cdot v^\parallel_{\omega}}{|B_{\omega}|}\times v^\perp_{\omega}$ is locally uniformly bounded it follows that the $C^{0,1}_{\operatorname{loc}}$-norm of $\int_0^t\frac{Db_{\omega}\cdot v^\parallel_{\omega}}{|B_{\omega}|}\times v^\perp_{\omega}ds$ is bounded so that along any subsequence we can find a further subsequence such that $\int_0^t\frac{Db_{\omega}\cdot v^\parallel_{\omega}}{|B_{\omega}|}\times v^\perp_{\omega}ds$ converges locally uniformly to some limit function. Due to the weak convergence to zero we conclude that $\int_0^t\frac{Db_{\omega}\cdot v^\parallel_{\omega}}{|B_{\omega}|}\times v^\perp_{\omega}ds$ converges pointwise to zero. Hence the full sequence $\int_0^t\frac{Db_{\omega}\cdot v^\parallel_{\omega}}{|B_{\omega}|}\times v^\perp_{\omega}ds$ converges locally uniformly to zero. Using this reasoning together with (\ref{5E12}) in (\ref{5E11}), yields
	\begin{gather}
		\nonumber
		\int_0^t\frac{Db_{\omega}\cdot v_{\omega}}{|B_{\omega}|}\times v_{\omega}ds=\int_0^t\frac{Db(x(s))\cdot v(s)}{|B(x(s))|}\times v(s)ds
		\\
		\label{5E13}
		+\int_0^t\frac{Db_{\omega}\cdot v^\perp_{\omega}}{|B_{\omega}|}\times v^\perp_{\omega}ds+o(1)\text{ in }C^0_{\operatorname{loc}}\text{ as }\omega\rightarrow\infty.
	\end{gather}
	We now recall (\ref{5E5}) and (\ref{5E6}) and write
	\begin{gather}
		\nonumber
		\int_0^t\frac{Db_{\omega}\cdot v^\perp_{\omega}}{|B_{\omega}|}\times v^\perp_{\omega}ds=\int_0^t\frac{Db_{\omega}\cdot (b_{\omega}\times (v_{\omega}\times b_{\omega}))}{|B_{\omega}|}\times (b_{\omega}\times (v_{\omega}\times b_{\omega}))ds
		\\
		\nonumber
		=\int_0^t\frac{Db_{\omega}\cdot (b_{\omega}\times (v_{\omega}\times b_{\omega}))}{|B_{\omega}|^2}\times \left(b_{\omega}\times\frac{\dot{v}_{\omega}}{\omega}\right)ds.
	\end{gather}
	We now integrate by parts and see that in the limit all terms vanish except for those where the derivative acts upon $v_{\omega}$, i.e.
	\begin{gather}
		\nonumber
		\int_0^t\frac{Db_{\omega}\cdot (b_{\omega}\times (v_{\omega}\times b_{\omega}))}{|B_{\omega}|^2}\times \left(b_{\omega}\times\frac{\dot{v}_{\omega}}{\omega}\right)ds
		\\
		\nonumber
		=-\int_0^t\frac{Db_{\omega}\cdot (b_{\omega}\times (\dot{v}_{\omega}\times b_{\omega}))}{|B_{\omega}|^2}\times \left(b_{\omega}\times\frac{v_{\omega}}{\omega}\right)ds+o(1)\text{ in }C^0_{\operatorname{loc}}\text{ as }\omega\rightarrow\infty.
	\end{gather}
	We note that $b_{\omega}\times (\dot{v}_{\omega}\times b_{\omega})=\dot{v}_{\omega}-(b_{\omega}\cdot \dot{v}_{\omega})b_{\omega}=\dot{v}_{\omega}=\omega v_{\omega}\times B_{\omega}$ since $\dot{v}_{\omega}$ is orthogonal to $b_{\omega}$. Hence
	\begin{gather}
		\nonumber
		-\int_0^t\frac{Db_{\omega}\cdot (b_{\omega}\times (\dot{v}_{\omega}\times b_{\omega}))}{|B_{\omega}|^2}\times \left(b_{\omega}\times\frac{v_{\omega}}{\omega}\right)ds=\int_0^t\frac{Db_{\omega}(v_{\omega}\times b_{\omega})}{|B_{\omega}|}\times (v_{\omega}\times b_{\omega})ds.
	\end{gather}
	We obtain
	\begin{gather}
		\label{5E14}
		\int_0^t\frac{Db_{\omega}\cdot v^\perp_{\omega}}{|B_{\omega}|}\times v^\perp_{\omega}ds=\int_0^t\frac{Db_{\omega}\cdot(v_{\omega}\times b_{\omega})}{|B_{\omega}|}\times (v_{\omega}\times b_{\omega})ds+o(1)\text{ in }C^0_{\operatorname{loc}}\text{ as }\omega\rightarrow\infty.
	\end{gather}
	Using that $v_{\omega}\times b_{\omega}=v^\perp_{\omega}\times b_{\omega}$ and \Cref{5L4} we find
	\begin{gather}
		\nonumber
		\int_0^t\frac{Db_{\omega}\cdot v^\perp_{\omega}}{|B_{\omega}|}\times v^\perp_{\omega}ds=-\int_0^t\frac{(|v_0|^2-h^2_{\omega})\operatorname{curl}^\parallel(b_{\omega})+(\nabla_{v^\perp_{\omega}}b_{\omega}\cdot (v_{\omega}\times b_{\omega}))b_{\omega}}{|B_{\omega}|}ds+o(1)
		\\
		\label{5E15}
		=\int_0^t\mathcal{O}(1)b(x_{\omega}(s))ds\text{ in }C^0_{\operatorname{loc}}\text{ as }\omega\rightarrow\infty
	\end{gather}
	where we used that $|v^\perp_{\omega}|^2=|v_{\omega}|^2-|v^\parallel_{\omega}|^2=|v_0|^2-h^2_{\omega}$, $h_{\omega}(t)=b(x_{\omega}(t))\cdot v_{\omega}(t)$, and we set $\operatorname{curl}^{\parallel}(b_{\omega}):=(\operatorname{curl}(b)(x_{\omega})\cdot b_{\omega})b_{\omega}$.
	\newline
	\newline
	We now turn to the remaining term in (\ref{S5Extra3}) and make use of (\ref{5E6}) in order to integrate by parts where all the terms where the derivative does not act on $v_{\omega}$ will be of order $o(1)$ in the limit
	\begin{gather}
		\nonumber
		\int_0^t\frac{\frac{\nabla |B_{\omega}(s)|^2}{2}\cdot v_{\omega}(s)}{|B_{\omega}(s)|^4}\left(B_{\omega}(s)\times v_{\omega}(s)\right)ds=-\frac{1}{\omega}\int_0^t\frac{\frac{\nabla |B_{\omega}(s)|^2}{2}\cdot v_{\omega}(s)}{|B_{\omega}(s)|^4}\dot{v}_{\omega}(s)ds
		\\
		\nonumber
		=\frac{1}{\omega}\int_0^t\frac{\frac{\nabla |B_{\omega}(s)|^2}{2}\cdot \dot{v}_{\omega}(s)}{|B_{\omega}(s)|^4}v_{\omega}(s)ds+o(1)=\int_0^t\frac{\frac{\nabla |B_{\omega}(s)|^2}{2}\cdot \left(v_{\omega}(s)\times B_{\omega}(s)\right)}{|B_{\omega}(s)|^4}v_{\omega}(s)ds+o(1)
		\\
		\nonumber
		=\int_0^t\frac{\frac{\nabla |B_{\omega}(s)|^2}{2}\cdot \left(v_{\omega}(s)\times B_{\omega}(s)\right)}{|B_{\omega}(s)|^4}v^\perp_{\omega}(s)ds+o(1)\text{ in }C^0_{\operatorname{loc}}\text{ as }\omega\rightarrow\infty.
	\end{gather}
	The last identity follows from the same arguments as in the derivation of (\ref{5E13}). We now make use of the vector calculus identity $\nabla \frac{|B|^2}{2}=\nabla_BB+B\times \operatorname{curl}(B)$ and observe that $\nabla_Bb=\frac{\nabla_BB}{|B|}+\left(B\cdot \nabla |B|^{-1}\right)B$ so that $\frac{\nabla_BB}{|B|^2}\cdot (v\times B)=\frac{\nabla_Bb}{|B|}\cdot (v\times B)=\nabla_bb\cdot (v\times B)$ for any $v\in \mathbb{R}^3$. With this we obtain
	\begin{gather}
		\nonumber
		\int_0^t\frac{\frac{\nabla |B_{\omega}(s)|^2}{2}\cdot v_{\omega}(s)}{|B_{\omega}(s)|^4}\left(B_{\omega}(s)\times v_{\omega}(s)\right)ds=\int_0^t\frac{\left(b_{\omega}\times \operatorname{curl}(B_{\omega})\right)\cdot (v_{\omega}\times b_{\omega})}{|B_{\omega}|^2}v^\perp_{\omega}ds
		\\
		\nonumber
		+\int_0^t\frac{\nabla_{b_{\omega}}b_{\omega}\cdot (v_{\omega}\times b_{\omega})}{|B_{\omega}|}v^\perp_{\omega}ds+o(1)\text{ in }C^0_{\operatorname{loc}}\text{ as }\omega\rightarrow\infty.
	\end{gather}
	We make once more use of our vector calculus identity $0=\nabla \frac{|b|^2}{2}=\nabla_bb+b\times \operatorname{curl}(b)$ since $b$ is a unit field. From this we conclude $\nabla_bb\cdot (v\times b)=\nabla_bb\cdot (v^\perp\times b)=(\operatorname{curl}(b)\times b)\cdot (v^\perp\times b)=(v^\perp\cdot \operatorname{curl}(b))-(b\cdot v^\perp)(b\cdot \operatorname{curl}(b))=v^\perp\cdot \operatorname{curl}(b)$ where $v^\perp=v-(v\cdot b)b$. Further, $(b\times \operatorname{curl}(B))\cdot (v\times b)=(b\times \operatorname{curl}(B))\cdot (v^\perp\times b)=(b\cdot v^\perp)(\operatorname{curl}(B)\cdot b)-\operatorname{curl}(B)\cdot v^\perp=-\operatorname{curl}(B)\cdot v^\perp$ and so we arrive at
	\begin{gather}
		\nonumber				
		\int_0^t\frac{\frac{\nabla |B_{\omega}(s)|^2}{2}\cdot v_{\omega}(s)}{|B_{\omega}(s)|^4}\left(B_{\omega}(s)\times v_{\omega}(s)\right)ds=-\int_0^t\frac{\operatorname{curl}(B_{\omega})\cdot v^\perp_{\omega}}{|B_{\omega}|^2}v^\perp_{\omega}ds
		\\
		\nonumber
		+\int_0^t\frac{\operatorname{curl}(b_{\omega})\cdot v^\perp_{\omega}}{|B_{\omega}|}v^\perp_{\omega}ds+o(1)\text{ in }C^0_{\operatorname{loc}}\text{ as }\omega\rightarrow\infty.
	\end{gather}
	We now make use of the formula $\operatorname{curl}(b)=\frac{\operatorname{curl}(B)}{|B|}+\nabla \left(|B|^{-1}\right)\times B$ so that we obtain
	\begin{gather}
		\label{5E17}
		\int_0^t\frac{\frac{\nabla |B_{\omega}(s)|^2}{2}\cdot v_{\omega}(s)}{|B_{\omega}(s)|^4}\left(B_{\omega}(s)\times v_{\omega}(s)\right)ds=\int_0^t\left((\nabla \left(|B_{\omega}|^{-1}\right)\times b_{\omega})\cdot v^\perp_{\omega}\right)v^\perp_{\omega}ds+o(1)\text{ in }C^0_{\operatorname{loc}}.
	\end{gather}
	We can now combine (\ref{5E13}),(\ref{5E15}) and (\ref{5E17})
	\begin{gather}
		\nonumber
		\int_0^t\frac{\frac{\nabla |B_{\omega}(s)|^2}{2}\cdot v_{\omega}(s)}{|B_{\omega}(s)|^4}(B_{\omega}(s)\times v_{\omega}(s))ds-\int_0^t\frac{Db_{\omega}(s)\cdot v_{\omega(s)}}{|B_{\omega}(s)|}\times v_{\omega}(s)ds
		\\
		\nonumber
		=-\int_0^t\frac{Db(x(s))\cdot v(s)}{|B(x(s))|}\times v(s)ds+\int_0^t\left(\left(\nabla \left(|B_{\omega}|^{-1}\right)\times b_{\omega}\right)\cdot v^\perp_{\omega}\right)v^\perp_{\omega}ds
		\\
		\label{5E18}
		+\int_0^t\mathcal{O}(1)b(x_{\omega}(s))ds+o(1)\text{ in }C^0_{\operatorname{loc}}\text{ as }\omega\rightarrow\infty.
	\end{gather}
	We now turn to the term in (\ref{5E18}) which involves $\nabla \left(|B_{\omega}|^{-1}\right)$. To simplify it we write $v^\perp_{\omega}=b_{\omega}\times (v_\omega\times b_{\omega})=\frac{b_{\omega}}{|B_{\omega}|}\times \frac{\dot{v}_{\omega}}{\omega}$ and integrate by parts which yields
	\begin{gather}
		\nonumber
		\int_0^t(v^\perp_{\omega}\cdot \left(\nabla \left(|B_{\omega}|^{-1}\right)\times b_{\omega}\right))v^\perp_{\omega}ds
		\\
		\label{5E19}
		=\int_0^t\left(\left(\nabla \left(|B_{\omega}|^{-1}\right)\times b_{\omega}\right)\cdot (v_{\omega}\times b_{\omega})\right)(v_{\omega}\times b_{\omega})ds+o(1)\text{ in }C^0_{\operatorname{loc}}\text{ as }\omega\rightarrow\infty.
	\end{gather}
	We recall that we had set $h_{\omega}=b(x_{\omega})\cdot v_{\omega}$ so that $|v_0|^2-h^2_{\omega}=|v^\perp_{\omega}|^2=|v_{\omega}|^2-(v_{\omega}\cdot b_{\omega})^2=|v_{\omega}\times b_{\omega}|^2$ and that $\{b_{\omega},v^{\perp}_{\omega},v_{\omega}\times b_{\omega}\}$ forms an orthogonal basis of $\mathbb{R}^3$. We can therefore express
	\begin{gather}
		\nonumber
		\nabla \left(|B_{\omega}|^{-1}\right)\times b_{\omega}=\frac{\left(\nabla \left(|B_{\omega}|^{-1}\right)\times b_{\omega}\right)\cdot v^\perp_{\omega}}{|v_0|^2-h^2_{\omega}}v^\perp_{\omega}+\frac{\left(\nabla \left(|B_{\omega}|^{-1}\right)\times b_{\omega}\right)\cdot (v_{\omega}\times b_{\omega})}{|v_0|^2-h^2_{\omega}}v_{\omega}\times b_{\omega}
	\end{gather}
	where we used that $\left(\nabla \left(|B_{\omega}|^{-1}\right)\times b_{\omega}\right)\cdot b_{\omega}=0$. We hence obtain from (\ref{5E19})
	\begin{gather}
		\nonumber
		\int_0^t(v^\perp_{\omega}\cdot \left(\nabla \left(|B_{\omega}|^{-1}\right)\times b_{\omega}\right))v^\perp_{\omega}ds=\frac{1}{2}\int_0^t(|v_0|^2-h^2_{\omega})\left(\nabla \left(|B_{\omega}|^{-1}\right)\times b_{\omega}\right)ds+o(1)\text{ in }C^0_{\operatorname{loc}}\text{ as }\omega\rightarrow\infty.
	\end{gather}
	We can insert this into (\ref{5E18})
	\begin{gather}
		\nonumber
		\int_0^t\frac{\frac{\nabla |B_{\omega}(s)|^2}{2}\cdot v_{\omega}(s)}{|B_{\omega}(s)|^4}(B_{\omega}(s)\times v_{\omega}(s))ds-\int_0^t\frac{Db_{\omega}(s)\cdot v_{\omega(s)}}{|B_{\omega}(s)|}\times v_{\omega}(s)ds
		\\
		\nonumber
		=-\int_0^t\frac{Db(x(s))\cdot v(s)}{|B(x(s))|}\times v(s)ds+\int_0^t\frac{|v_0|^2-h^2_{\omega}}{2}\nabla \left(|B_{\omega}|^{-1}\right)\times b_{\omega}ds
		\\
		\nonumber
		+\int_0^t\mathcal{O}(1)b_{\omega}(s)ds+o(1)\text{ in }C^0_{\operatorname{loc}}\text{ as }\omega\rightarrow\infty.
	\end{gather}
	We observe that the corresponding integrands now all converge locally uniformly and so we are allowed to perform the limit. Hence we finally arrive at
	\begin{gather}
		\nonumber
		\int_0^t\frac{\frac{\nabla |B_{\omega}(s)|^2}{2}\cdot v_{\omega}(s)}{|B_{\omega}(s)|^4}(B_{\omega}(s)\times v_{\omega}(s))ds-\int_0^t\frac{Db_{\omega}(s)\cdot v_{\omega(s)}}{|B_{\omega}(s)|}\times v_{\omega}(s)ds
		\\
		\nonumber
		=-\int_0^t\frac{Db(x(s))\cdot v(s)}{|B(x(s))|}\times v(s)ds+\int_0^t\frac{|v_0|^2-h^2(s)}{2}\nabla \left(|B|^{-1}\right)(x(s))\times b(x(s))ds
		\\
		\label{5E20}
		+\int_0^t\mathcal{O}(1)b(x_{\omega}(s))ds+o(1)\text{ in }C^0_{\operatorname{loc}}\text{ as }\omega\rightarrow\infty.
	\end{gather}
	We will now have a closer look at the first term in (\ref{5E20}). To this end we recall that $v=\dot{x}$ is the derivative of the zero order approximation $x$ which satisfies the identity $v(s)=h(s)b(x(s))$, see \Cref{S2T1}, so that $(Db \cdot v)\times v=h^2 (Db\cdot b)\times b$. Further, we have $Db\cdot b=\nabla_bb=\kappa$ (recall this is precisely the definition of the magnetic curvature $\kappa$) and so $(Db\cdot b)\times b=\kappa\times b$.
	
	In addition, according to the chain rule, we have
	\begin{gather}
		\nonumber
		\frac{|v_0|^2-h^2(s)}{2}\nabla \left(|B|^{-1}\right)(x(s))\times b(x(s))=\frac{|v_0|^2-h^2(s)}{2|B(x(s))|}\frac{b(x(s))\times\nabla |B|(x(s))}{|B(x(s))|}.
	\end{gather}
	We notice, recall \Cref{S2T1}, that $\frac{|v_0|^2-h^2(s)}{2|B(x(s))|}=\mu(t)=\mu_0$ where we used that the magnetic moment is preserved in time, cf. \Cref{S2T1}. With these considerations (\ref{5E20}) becomes
	\begin{gather}
		\nonumber
		\int_0^t\frac{\frac{\nabla |B_{\omega}(s)|^2}{2}\cdot v_{\omega}(s)}{|B_{\omega}(s)|^4}(B_{\omega}(s)\times v_{\omega}(s))ds-\int_0^t\frac{Db_{\omega}(s)\cdot v_{\omega(s)}}{|B_{\omega}(s)|}\times v_{\omega}(s)ds
		\\
		\nonumber
		=\int_0^t\frac{h^2(s)}{|B(x(s))|}b(x(s))\times \kappa(x(s))ds+\mu_0\int_0^t\frac{b(x(s))\times \nabla \left(|B|\right)(x(s))}{|B(x(s))|}ds
		\\
		\label{S5Extra4}
		+\int_0^t\mathcal{O}(1)b(x_{\omega}(s))ds+o(1)\text{ in }C^0_{\operatorname{loc}}\text{ as }\omega\rightarrow\infty.
	\end{gather}
	The theorem now follows from (\ref{S5Extra1}),(\ref{S5Extra2}), (\ref{S5Extra3}) and (\ref{S5Extra4}).
\end{proof}
\subsection{Guiding centre expansion}
\begin{proof}[Proof of \Cref{S2T4}]
	We start with the following formula which we derived during the proof of \Cref{S2T2}, recall (\ref{S5Extra1}), (\ref{S5Extra2}),(\ref{S5Extra3}),(\ref{S5Extra4})
	\begin{gather}
		\nonumber
		\Scale[1.1]{x_{\omega}(t)=x_0+\int_0^t(h(s)+o(1))b(x_{\omega}(s))ds+\frac{b(x(t))\times v_{\omega}(t)}{\omega|B(x(t))|}-\frac{b(x_0)\times v_0}{\omega|B(x_0)|}+\frac{\int_0^t\frac{h^2(s)}{|B(x(s))|}b(x(s))\times \kappa(x(s))ds}{\omega}}
		\\
		\label{S5E23}
		+\frac{\mu_0}{\omega}\int_0^t\frac{b(x(s))\times \nabla |B|(x(s))}{|B(x(s))|}ds
		+\int_0^t\mathcal{O}\left(\frac{1}{\omega}\right)b(x_{\omega}(s))ds+o\left(\frac{1}{\omega}\right)\text{ in }C^0_{\operatorname{loc}}\text{ as }\omega\rightarrow\infty
	\end{gather}
	where $\kappa(x):=(\nabla_bb)(x)$, $b(x):=\frac{B(x)}{|B(x)|}$ and $\mu_0\equiv\mu(0),h(s),x(s)$ are as in \Cref{S2T1}.
	
	We notice first that $x_{\omega}(t)\rightarrow x(t)$ in $C^0_{\operatorname{loc}}$, cf. \Cref{S2T1}. In addition we recall that (\ref{S1E2}) implies that $|v_{\omega}(t)|=|v_0|$ so that we immediately obtain
	\begin{gather}
		\nonumber
		\rho_{\omega}(t)=\frac{b(x_{\omega}(t))\times v_{\omega}(t)}{\omega |B(x_{\omega}(t))|}=\frac{b(x(t))\times v_{\omega}(t)}{\omega |B(x(t))|}+o\left(\frac{1}{\omega}\right)\text{ in }C^0_{\operatorname{loc}}\text{ as }\omega\rightarrow\infty
	\end{gather}
	where $\rho_{\omega}(t)$ denotes the gyromotion. According to \Cref{S2D3} we have
	\begin{gather}
		\nonumber
		R_{\omega}(t)=x_{\omega}(t)-\rho_{\omega}(t).
	\end{gather}
	Since $\rho_{\omega}\in \mathcal{O}\left(\frac{1}{\omega}\right)$ in $C^0_{\operatorname{loc}}$ we conclude in particular that $R_{\omega}(t)\rightarrow x(t)$ in $C^0_{\operatorname{loc}}$ as $\omega\rightarrow\infty$. Using these observations we deduce from (\ref{S5E23})
	\begin{gather}
		\nonumber
		R_{\omega}(t)=x_0-\frac{b(x_0)\times v_0}{\omega|B(x_0)|}+\int_0^t(h(s)+o(1))b(R_{\omega}(s))ds
		\\
		\nonumber
		+\frac{\int_0^t\frac{h^2(s)}{|B(R_{\omega}(s))|}b(R_{\omega}(s))\times \kappa(R_{\omega}(s))ds}{\omega}+\frac{\mu_0}{\omega}\int_0^t\frac{b(R_{\omega}(s))\times \nabla |B|(R_{\omega}(s))}{|B(R_{\omega}(s))|}ds
		\\
		\label{S5E24}
		+\int_0^th(s) (b(x_{\omega}(s))-b(R_{\omega}(s)))ds+o\left(\frac{1}{\omega}\right)\text{ in }C^0_{\operatorname{loc}}\text{ as }\omega\rightarrow\infty.
	\end{gather}
	We observe that a priori $x_{\omega}(s)-R_{\omega}(s)=\rho_{\omega}(s)\in \mathcal{O}\left(\frac{1}{\omega}\right)$ and so the term containing $b(x_{\omega}(s))-b(R_{\omega}(s))$ might contribute a term of order $\frac{1}{\omega}$. We claim however that
	\begin{gather}
		\label{S5E25}
		\int_0^th(s)\left(b(x_{\omega}(s))-b(R_{\omega}(s))\right)ds\in o\left(\frac{1}{\omega}\right)\text{ in }C^0_{\operatorname{loc}}\text{ as }\omega\rightarrow\infty.
	\end{gather}
	Notice that (\ref{S5E25}) in combination with (\ref{S5E24}) yields precisely the statement of \Cref{S2T4} so that we are left with establishing (\ref{S5E25}).
	
	To this end we use the following identity
	\begin{gather}
		\nonumber
		b(y)-b(z)=\int_0^1\frac{d}{d\tau}\left(b(z+\tau(y-z))\right)d\tau=\int_0^1(Db)(z+\tau(y-z))d\tau\cdot (y-z)\text{ for all }y,z\in \mathbb{R}^3.
	\end{gather}
	From this we conclude
	\begin{gather}
		\nonumber
		b(x_{\omega}(s))-b(R_{\omega}(s))=\int_0^1(Db)(R_{\omega}(s)+\tau(x_{\omega}(s))-R_{\omega}(s))d\tau\cdot (x_{\omega}(s)-R_{\omega}(s)).
	\end{gather}
	We notice now that $R_{\omega}(s)+\tau (x_{\omega}(s)-R_{\omega}(s))\rightarrow x(s)$ locally uniformly in $\tau$ and $s$ and that $x_{\omega}(s)-R_{\omega}(s)=\rho_{\omega}(s)\in \mathcal{O}\left(\frac{1}{\omega}\right)$ in $C^0_{\operatorname{loc}}$ from which we overall conclude
	\begin{gather}
		\label{S5E26}
		b(x_{\omega}(s))-b(R_{\omega}(s))=(Db)(x(s))\cdot \rho_{\omega}(s)+o\left(\frac{1}{\omega}\right)\text{ in }C^0_{\operatorname{loc}}\text{ as }\omega\rightarrow\infty.
	\end{gather}
	We can insert (\ref{S5E26}) into the left hand side of (\ref{S5E25}) to deduce
	\begin{gather}
		\nonumber
		\int_0^th(s)(b(x_{\omega}(s))-b(R_{\omega}(s)))ds=-\int_0^th(s)(Db)(x(s))\cdot \frac{v_{\omega}(s)\times b(x_{\omega}(s))}{\omega|B(x_{\omega}(s))|}ds+o\left(\frac{1}{\omega}\right) 
		\\
		\label{S5E27}
		=-\int_0^th(s)(Db)(x(s))\cdot \frac{v_{\omega}(s)\times B(x_{\omega}(s))}{\omega|B(x(s))|^2}ds+o\left(\frac{1}{\omega}\right)\text{ in }C^0_{\operatorname{loc}}\text{ as }\omega\rightarrow\infty
	\end{gather}
	where we used that $|v_{\omega}|$ is uniformly bounded and that $x_{\omega}(t)\rightarrow x(t)$ in $C^0_{\operatorname{loc}}$ in the last step. We recall (\ref{S1E2}) which states that $\dot{v}_{\omega}(t)=\omega (v_{\omega}(t)\times B(x_{\omega}(t)))$. We can insert this identity into (\ref{S5E27}) and integrate by parts
	\begin{gather}
		\nonumber
			\int_0^th(s)(b(x_{\omega}(s))-b(R_{\omega}(s)))ds=-\int_0^th(s)(Db)(x(s))\cdot \frac{\dot{v}_{\omega}(s)}{\omega^2|B(x(s))|^2}ds+o\left(\frac{1}{\omega}\right)
			\\
			\label{S5E28}
			=-h(s)(Db)(x(s))\cdot \frac{v_{\omega}(s)}{\omega^2|B(x(s))|^2}\bigg|_0^t+\int_0^t\frac{\frac{d}{ds}\left(\frac{h(s)(Db)(x(s))}{|B(x(s))|^2}\right)}{\omega^2}\cdot v_{\omega}(s)ds+o\left(\frac{1}{\omega}\right)\text{ in }C^0_{\operatorname{loc}}\text{ as }\omega\rightarrow\infty.
	\end{gather}
	We recall that $|v_{\omega}(t)|=|v_0|$ for all $\omega,t$ so that the final expression in (\ref{S5E28}) is of order $o\left(\frac{1}{\omega}\right)$ in $C^0_{\operatorname{loc}}$ as $\omega\rightarrow\infty$ which establishes (\ref{S5E25}). As mentioned previously (\ref{S5E24}) in combination with (\ref{S5E25}) then establishes the theorem.
\end{proof}

\section{Physics literature approach\label{s-physics}}
\subsection{Some definitions}
We follow here the exposition in \cite[Chapter 4]{IGPW24}, see also \cite[Chapter 5]{IGPW20}, since it gives a modern presentation with clearly stated assumptions on the quantities involved under which a corresponding guiding centre expansion is derived.
\newline
\newline
For given initial conditions $x_0,v_0\in \mathbb{R}^3$ we denote the exact particle position, i.e. the solution to (\ref{S1E1}), by $\vec{r}$. Then the particle motion $\vec{r}$ is decomposed into a guiding centre motion $\vec{R}$ and a gyromotion $\vec{\rho}$ with $\vec{r}=\vec{R}+\vec{\rho}$ where $\vec{\rho}(t)=\rho(t)\vec{u}(t)$ for some unit vector $\vec{u}$ which lies in the plane perpendicular to the magnetic field $B$ at the guiding centre $\vec{R}(t)$ and $\rho$ being the length of $\vec{\rho}$, commonly referred to as the gyroradius, cf. \Cref{Figure1}.

	\begin{figure}[H]
	\centering
	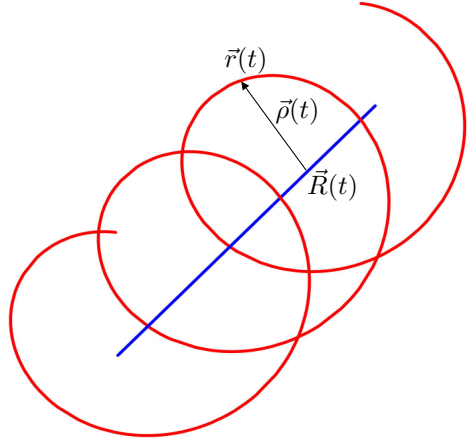
	\caption{The real particle position is depicted in red. The "dominant" guiding centre motion is depicted in blue. The gyromotion is depicted in black. As we follow the trajectory, the gyro motion rotates at some angle $\phi$.}
	\label{Figure1}
\end{figure}

We let $\hat{b}$ denote the normalised magnetic field, i.e. $\hat{b}=b$ in the notation from \Cref{S2T4}. We use this new notation to be consistent with the physics literature notation from \cite[Chapter 5]{IGPW20} in the present section. We further adapt the following notation from \cite[Chapter 5]{IGPW20} (here $\bold{B}$ is the non-normalised, i.e. physical, magnetic field as in (\ref{S1E1})):
\begin{enumerate}
	\item $V_{\parallel}$ denotes the (scalar) component of the guiding centre $\vec{R}$ along $\hat{b}$, i.e. $V_{\parallel}=\hat{b}(\vec{R})\cdot \dot{\vec{R}}$.
	\item We define $\Omega(\vec{R}):=\frac{q}{m} |\bold{B}(\vec{R})|$ where $m,q$ are the mass and charge of the particle.
	\item $v_{\perp}:=\rho \Omega(\vec{R})$ the perpendicular velocity
	\item $\mu:=\frac{v^2_{\perp}}{2|\bold{B}(\vec{R})|}$ the magnetic moment. This quantity is referred to as an adiabatic invariant.
	\item $E:=v^2_{\perp}+V^2_{\parallel}$ the kinetic energy (note that we omit here a factor of $\frac{m}{2}$, $m$ being the mass of the particle) where we set the electric potential $\Phi=0$ since we consider here only the situation of a moving particle in a static magnetic field. This is the conserved energy corresponding to the averaged Lagrangian in the physics literature approach.
	\item $\kappa:=\nabla_{\hat{b}}\hat{b}$ where $\nabla_XY$ denotes the covariant derivative of $Y$ w.r.t. $X$ along the Euclidean Levi-Civita connection, i.e. in Cartesian coordinates $\kappa=\hat{b}^j(\partial_j\hat{b}^i)e_i$ where the $e_i$ denote the standard basis vectors. This quantity is known as the magnetic curvature.
\end{enumerate}
\subsection{Smallness assumptions}
To derive an appropriate expansion for the guiding centre motion the following assumptions are made \cite[Chapter 4]{IGPW24},\cite[Chapter 5]{IGPW20}
\begin{enumerate}
	\item $\frac{\rho}{L_B}\ll 1$ where $\rho$ is the gyroradius and $L_B:=\frac{|\bold{B}|}{|\nabla |\bold{B}||}$ is the length scale of the spatial variation of the magnetic field strength.
	\item It is assumed that there is a suitable orthonormal moving frame along the guiding centre motion such that the rotation angle $\phi$ of the gyromotion with respect to this frame satisfies the condition $\frac{\omega_B}{\dot{\phi}}\ll 1$ where $\dot{\phi}$ is called the gyrofrequency of the system and where $\omega_B\sim \frac{v_t}{L_B}$ is the frequency associated with the length scale of the spatial variation of the magnetic field and the associated thermal velocity $v_t=\sqrt{\frac{2T}{m}}$ of the charged particle (where we omitted the Boltzmann constant in accordance with \cite{IGPW20},\cite{IGPW24}). As a consequence of the assumptions made here and on the structure of the solution below one can derive the relationship $\dot{\phi}\sim\Omega$ with $\Omega$ as above. Due to this, $\Omega$ is also sometimes referred to as the gyrofrequency. A more rigorous interpretation of the angle $\phi$ is discussed in \Cref{MathStructure}, bullet point (iii); see also (\ref{Angle}) for an explicit formula.
	\item An additional assumption on the electric potential is made which we omit here, since we consider the situation of an absent electric field.
\end{enumerate}
\subsection{Structural assumptions}
The following additional assumptions are made on the \textit{structure} of the guiding centre and gyromotions (motivated by the situation of a uniform magnetic field)
\begin{enumerate}
	\item $V_{\parallel}=\dot{\vec{R}}\cdot \hat{b}(\vec{R})\sim v_t$.
	\item For $\vec{V}^\perp\equiv \dot{\vec{R}}^{\perp}=\dot{\vec{R}}-V_{\parallel}\hat{b}(\vec{R})$ we have $|\vec{V}^\perp|\sim \epsilon v_t$ where $\epsilon\sim \frac{\rho}{L_B}\sim \frac{\omega_B}{\dot{\phi}}$ is the "smallnes" scale, i.e. the drift of the guiding centre perpendicular to the magnetic field is assumed to be small in comparison to its motion parallel to the field.
	\item $\dot{\phi}\sim \frac{q}{m} |\bold{B}|$, i.e. the gyrofrequency is assumed to be of the same order of magnitude as $\Omega(\vec{R})$.
	\item $\dot{\rho}\sim \omega_B \rho$ where $\rho$ denotes the (scalar) gyroradius.
\end{enumerate}
\subsection{The guiding centre equations of motion}
The leading order equations of motion which are obtained in the physics literature under the above smallness and structural assumptions are the following, c.f. \cite{IGPW20},\cite{IGPW24}
\begin{gather}
	\label{S3E1}
	\dot{V}_{\parallel}=-\mu \hat{b}(\vec{R})\cdot \nabla |B|(\vec{R})\text{ (parallel guiding centre motion with zero electric field)},
	\\
	\label{S3E2}
	\dot{\vec{R}}^\perp=\frac{V^2_{\parallel}}{\Omega(\vec{R})}\hat{b}(\vec{R})\times \kappa(\vec{R})+\frac{\mu}{\Omega(\vec{R})}\hat{b}(\vec{R})\times \nabla |B|(\vec{R})\text{ (perpendicular guiding centre motion)}.
\end{gather}
\section{Comparison of the mathematical and physics literature results}\label{s-physics-math}
\subsection{Reconciling the mathematical and physical approach}
The notations used in \Cref{S2D3} and \Cref{S2T4} are already suggestive and in the following we discuss their properties with a focus on their physical counter parts.
\begin{enumerate}
	\item We identify the mathematical guiding centre position $R_{\omega}(t)$ with the physical guiding centre motion $\vec{R}(t)$ and the mathematical gyromotion $\rho_{\omega}(t)$ with the physical gyromotion $\vec{\rho}$. In the following we explain that with this identification all physical formulas are recovered from the corresponding mathematical formula (\ref{S2E4}).
	\item (\textit{Parallel guiding centre motion}) We note that according to \Cref{S2T4} the leading order term in the motion of the guiding centre parallel to $b$ is given by the term $\int_0^th(s)b(R_{\omega}(s))ds$ and consequently we can identify $h(t)$ with the corresponding leading order parallel motion $V_{\parallel}$ in the physics literature approach. It follows from \cite[Equation (2.7)]{BG25} that we have the identity
	\begin{gather}
		\label{S4E1}
		\dot{h}(t)=-\mu_0b(x(t))\cdot \nabla |B|(x(t))=-\mu_0 b(R_{\omega}(t))\cdot \nabla |B|(R_{\omega}(t))+o(1)\text{ as }\omega\rightarrow\infty
	\end{gather}
	where we used in the last step that $x_{\omega}(t)-x(t)=o(1)$ since $x_{\omega}(t)$ converges locally uniformly to $x(t)$ and that $x_{\omega}(t)-R_{\omega}(t)=\rho_{\omega}(t)\in \mathcal{O}\left(\frac{1}{\omega}\right)$ since $|v_{\omega}(t)|=|v_0|$ for all $t,\omega$.
	
	According to our identification we have the correspondence $V_{\parallel}\leftrightarrow h$ so that (\ref{S4E1}) coincides precisely with the physics literature formula (\ref{S3E1}) upon discarding higher order terms.
	\item (\textit{Perpendicular guiding centre motion}) We recall that the physical quantity $\dot{\vec{R}}^{\perp}$ appearing in (\ref{S3E2}) corresponds to the leading order (vectorial) motion of the guiding centre in perpendicular direction to the magnetic field. We observe that according to \Cref{S2T4} the corresponding mathematical motion of $R_{\omega}(t)$ in normal direction of the magnetic field is given by
	\begin{gather}
		\label{S4E2}
		\mu_0 \frac{b(R_{\omega}(t))\times \nabla |B|(R_{\omega}(t))}{\omega|B(R_{\omega}(s))|}+h^2(t)\frac{b(R_{\omega}(t))\times \kappa(R_{\omega}(t))}{\omega|B(R_{\omega}(s))|}+o\left(\frac{1}{\omega}\right)\text{ as }\omega\rightarrow\infty.
	\end{gather}
	We are therefore led to identify $\dot{\vec{R}}^{\perp}$ with the leading order term in (\ref{S4E2}). We observe that (\ref{S4E2}) does not yet exactly coincide with the corresponding physics literature formula (\ref{S3E2}).
	
	To make the identification complete, we recall that (\ref{S1E2}) is obtained from (\ref{S1E1}) by specifying some reference magnetic field strength $B_{\operatorname{ref}}$ and setting $B=\frac{\bold{B}}{B_{\operatorname{ref}}}$ and $\omega=\frac{qB_{\operatorname{ref}}}{m}$. We hence see that $\omega |B(y)|=\frac{q}{m}|\bold{B}(y)|$ and in particular
	\begin{gather}
		\nonumber
		\frac{1}{\omega |B(x(t))|}=\frac{1}{\omega |B(R_{\omega}(t))|}+o\left(\frac{1}{\omega}\right)=\frac{1}{\frac{q}{m}|\bold{B}(R_{\omega}(t))|}+o\left(\frac{1}{\omega}\right)\text{ as }\omega\rightarrow\infty
	\end{gather}
	where we again used that $R_{\omega}(t)-x(t)=o(1)$ as $\omega\rightarrow\infty$.
	\newline
	\newline
	Following the physics literature approach and defining $\Omega(R_{\omega}(t)):=\frac{q}{m}|\bold{B}(R_{\omega}(t))|=\omega |B(R_{\omega}(t))|$ we therefore see that (\ref{S4E2}) becomes
	\begin{gather}
		\label{S4E3}
		\mu_0\frac{b(R_{\omega}(t))\times \nabla |B|(R_{\omega}(t))}{\Omega(R_{\omega}(t))}+h^2(t)\frac{b(R_{\omega}(t))\times \kappa(R_{\omega}(t))}{\Omega(R_{\omega}(t))}+o\left(\frac{1}{\omega}\right)\text{ as }\omega\rightarrow\infty.
	\end{gather}
	We recall that we have already identified the correspondence $V_{\parallel}\leftrightarrow h$ and that $\dot{\Vec{R}}^{\perp}$ corresponds to the leading order term in (\ref{S4E3}). Hence (\ref{S4E3}) coincides precisely with (\ref{S3E2}) upon using the identification $V_{\parallel}\leftrightarrow h$ and discarding higher order terms.
	\item In the physics literature setting we had the identities $\mu=\frac{v^2_{\perp}}{2|B(\vec{R})|}$ where $v_{\perp}=\rho\Omega(\vec{R})$ and $\rho$ is the length of the gyromotion $\vec{\rho}$. 
	
	We argue now that with our identifications the corresponding mathematical expression coincides up to higher order terms with $\mu_0$, the magnetic moment of the zero order limit trajectory as introduced in \Cref{S2T1}. This in particular justifies calling the physical magnetic moment an adiabatic invariant.
	\newline
	\newline
	To see this we recall that according to \Cref{S2D3} the gyromotion is given by
	\begin{gather}
		\nonumber
		\rho_{\omega}(t)=\frac{b(x_{\omega}(t))\times v_{\omega}(t)}{\omega|B(x_{\omega}(t))|}=\frac{b(R_{\omega}(t))\times v_{\omega}(t)}{\Omega(R_{\omega}(t))}+o\left(\frac{1}{\omega}\right)\text{ as }\omega\rightarrow\infty
	\end{gather}
	where we used that the $v_{\omega}$ is uniformly bounded. This shows that the leading order gyromotion is a motion which is perpendicular to $b(R_{\omega}(t))$ and thus is compatible with the decomposition of the physics literature. In addition, we compute
	\begin{gather}
		\nonumber
		\hat{\rho}_{\omega}(t):=|\rho_{\omega}(t)|=\frac{\sqrt{|v_{\omega}|^2-(b(x_{\omega})\cdot v_{\omega}(t))^2}}{\omega|B(x_{\omega}(t))|}+o\left(\frac{1}{\omega}\right)=\frac{\sqrt{|v_0|^2-h^2(t)}}{\Omega(R_{\omega}(t))}+o\left(\frac{1}{\omega}\right)\text{ as }\omega\rightarrow\infty,
	\end{gather}
	where we used that $|v_{\omega}(t)|=|v_0|$ for all $t,\omega$ and that $b(x_{\omega})\cdot v_{\omega}(t)\rightarrow h(t)$ locally uniformly, cf. \cite[Proof of Theorem 1.2]{BG25}. According to our identifications and the physics literature definition of the magnetic moment we compute
	\begin{gather}
		\nonumber
		\mu_{\operatorname{phys}}=\frac{\hat{\rho}^2_{\omega}(t)\Omega^2(R_{\omega}(t))}{2|B(R_{\omega}(t))|}=\frac{|v_0|^2-h^2(t)}{2|B(R_{\omega}(t))|}+o\left(1\right)
		\\
		\label{S4E4}
		=\frac{|v_0|^2-h^2(t)}{2|B(x(t))|}+o(1)=\mu_{\operatorname{math}}(t)+o(1)=\mu_0+o(1)\text{ as }\omega\rightarrow\infty
	\end{gather}
	where we used \Cref{S2T1} which states that the mathematical magnetic moment is preserved in time. We conclude that with our identification the mathematical definition of the magnetic moment, cf. \Cref{S2T1}, and the physical definition coincide up to higher order terms.
	\item We recall further that in the physics literature case the kinetic energy $E=v^2_{\perp}+V^2_{\parallel}$ was preserved in time. In our situation we have already identified the correspondence $V_{\parallel}\leftrightarrow h(t)$. According to the physics literature definition and our identification we additionally have the correspondence $v_{\perp}\leftrightarrow \hat{\rho}_{\omega}(t) \Omega(R_{\omega}(t))$. As in the derivation of (\ref{S4E4}) we can compute
	\begin{gather}
		\nonumber
		\left(\hat{\rho}_{\omega}(t) \Omega(R_{\omega}(t))\right)^2=|v_0|^2-h^2(t)+o(1)\text{ as }\omega\rightarrow\infty.
	\end{gather}
	Consequently we obtain
	\begin{gather}
		\nonumber
		E\leftrightarrow h^2(t)+v^2_0-h^2(t)+o(1)=|v_0|^2+o(1)\text{ as }\omega\rightarrow\infty
	\end{gather}
	and so the corresponding quantity $E$ is in leading order preserved in time.
	\item We lastly point out that in the physics literature approach \cite[Chapter 4]{IGPW24},\cite[Chapter 5]{IGPW20}, the equations of motion for the guiding centre are obtained by an averaging procedure which suggests that in some appropriate averaged sense the gyromotion should not contribute to the main guiding centre motion. With our identification we find that the mathematical gyromotion is given by
	\begin{gather}
		\nonumber
		\rho_{\omega}(t)=\frac{b(x_{\omega}(t))\times v_{\omega}(t)}{\omega |B(x_{\omega}(t))|}.
	\end{gather}
	Using the defining equation (\ref{S1E2}) it is easy to see by means of an integration by parts and the fact that $|v_{\omega}(t)|$ is uniformly bounded (and so in particular also $x_{\omega}(t)$ is locally uniformly bounded) that for every fixed $T>0$ we have
	\begin{gather}
		\label{S4E5}
		\frac{1}{T}\int_0^T\rho_{\omega}(t)dt=o\left(\frac{1}{\omega}\right)\text{ in }C^0_{\operatorname{loc}}(\mathbb{R})\text{ as }\omega\rightarrow\infty.
	\end{gather}
	Hence (\ref{S4E5}) shows that the time average $\frac{1}{T}\int_0^Tx_{\omega}(t)dt$ of the full particle trajectory $x_{\omega}(t)$ coincides up to higher order terms with the time average $\frac{1}{T}\int_0^TR_{\omega}(t)dt$ of the guiding centre motion. In this sense, the mathematical gyromotion can be understood as an oscillation whose contributions cancel each other upon averaging.  
\end{enumerate}
\subsection{Comparing the smallness assumptions}
The two assumptions of relevance in our context in the physics literature approach were the following
\begin{gather}
	\nonumber
	\frac{\rho}{L_B}\ll 1\text{ and }\frac{\omega_B}{\dot{\phi}}\ll 1
\end{gather}
where $\rho$ is the gyroradius, $L_B=\frac{|\bold{B}|}{|\nabla |\bold{B}||}$, $\omega_B\sim \frac{v_t}{L_B}$ with $v_t=\sqrt{\frac{2T}{m}}$ and (taking into account the structural assumptions) $\dot{\phi}\sim \frac{q}{m} |\bold{B}|=\omega |B|$ for any choice of reference magnetic field strength, where $\omega=\frac{qB_{\operatorname{ref}}}{m}$ and $B=\frac{\bold{B}}{B_{\operatorname{ref}}}$. In particular, the second physical smallness assumption corresponds to
\begin{gather}
	\label{S4E6}
	\frac{|\nabla |B||}{\omega|B|^2}\sqrt{\frac{2T}{m}}\ll 1.
\end{gather}
In the mathematical approach the assumption we made was
\begin{gather}
	\label{S4E7}
	\omega \gg 1.
\end{gather}
We see that for some fixed magnetic field and values of $T,m$, condition (\ref{S4E7}) implies condition (\ref{S4E6}). In contrast, it may be that $|B|$ is almost constant and therefore (\ref{S4E6}) may be satisfied even in the situation $\omega\sim 1$, i.e. when condition (\ref{S4E7}) is violated. However, it is clear that strong magnetic fields are necessary in order to confine hot plasmas in a magnetic confinement fusion device and therefore the reference magnetic field strength should reflect this property if we wish to obtain good approximations (a good choice would be to pick the reference strength as the average magnetic field strength or minimal magnetic field strength in some region of interest). From this perspective the condition $\omega\gg 1$ is very natural and mathematically also much simpler than condition (\ref{S4E6}).
\newline
\newline
Let us now examine the remaining physical smallness assumption
\begin{gather}
	\nonumber
	\frac{\rho}{L_B}\ll 1.
\end{gather}
With our identification $\rho\leftrightarrow \hat{\rho}_{\omega}(t)$ and we find the identity $L_B=\frac{|B|}{|\nabla |B||}$ since $L_B$ is a relative measure. This physical smallness condition translates to
\begin{gather}
	\label{S4E8}
	\frac{\hat{\rho}_{\omega}(t)}{L_B}\ll 1.
\end{gather}
We have already previously seen that $\hat{\rho}_{\omega}(t)=|\rho_{\omega}(t)|=\frac{\sqrt{|v_0|^2-h^2(t)}}{|B(x(t))|\omega}+o\left(\frac{1}{\omega}\right)$ as $\omega\rightarrow\infty$ where $x(t)$ denotes the zero order limit. In addition, since $R_{\omega}(t)-x(t)=o(1)$ as $\omega\rightarrow\infty$, we obtain $L_B=\frac{|B(R_{\omega}(t))|}{|\nabla |B|(R_{\omega}(t))|}=\frac{|B(x(t))|}{|\nabla |B|(x(t))|}+o(1)$ as $\omega\rightarrow\infty$. Consequently we find
\begin{gather}
	\label{S4E9}
	\frac{\hat{\rho}_{\omega}(t)}{L_B}=\frac{\sqrt{|v_0|^2-h^2(t)}|\nabla |B|(x(t))|}{|B(x(t))|^2\omega}+o\left(\frac{1}{\omega}\right)=o(1)\text{ as }\omega\rightarrow\infty.
\end{gather}
We conclude from (\ref{S4E9}) that if $\omega\gg 1$, then condition (\ref{S4E8}) is satisfied.

Most notably in the mathematical approach, in contrast to the physics literature approach, we do not need to assume a priori that the gyroradius is small in comparison to some characteristic length scale of the system, cf. (\ref{S4E8}). Instead, it turns out a posteriori that the gyroradius is of order $\frac{1}{\omega}$ and hence necessarily small in the strong magnetic field regime. We also emphasise that the definition of the gyromotion as given in \Cref{S2D3} is natural in view of (\ref{S4E5}) since it is defined precisely as the part of the leading order expansion which on average does not contribute to the particle position. 
\subsection{Mathematical structural consequences vs. physical structural assumptions}
\label{MathStructure}
In the physics literature approach a few structural assumptions have been made. In strong contrast to that no such assumptions have been made in the mathematicians approach. One can therefore try to see to what extent the structural assumptions in the physics literature approach are an a posteriori consequence in the mathematical approach and were thus justified to be made.
\begin{enumerate}
	\item The first assumption that has been made was $V_{\parallel}\sim v_t$. From a statistical point of view this can be justified by assuming a Maxwellian distribution of velocity and computing the average velocity (and its variance if desired) which shows that on average the velocity speed in any fixed direction coincides with $\sqrt{\frac{k_BT}{m}}=\frac{v_t}{\sqrt{2}}$ where $k_B$ is the Boltzmann constant \cite[Chapter 1.3]{Chen16}.
	
	Using our correspondence between the physical and mathematical quantitie
s we recall that $V_{\parallel}\leftrightarrow h(t)$. Hence the physical assumption would be equivalent to demanding that $h(t)\sim v_t$ where $v_t$ denotes the thermal velocity. However, there is no reason why $h(t)$ should be bounded away from zero and in fact it may change sign which would correspond to a "reflection" of the particle and physically describes a trapped particle as in a magnetic mirror. Near the reflection points the condition $V_{\parallel}\leftrightarrow h\sim v_t$ is violated and therefore the expansion derived by means of the physics literature approach is no longer justified. On the other hand, our mathematical approach solely requires the assumption $\omega \gg 1$ or equivalently that the magnetic field is strong. Consequently our expansion remains valid even in the regime $V_{\parallel}\ll v_t$. Hence the classical guiding centre approximation may still be used to describe the guiding centre behaviour near reflection points.
	\item An additional assumption that was made was $|\vec{V}^{\perp}|\sim \epsilon v_t$ where $\epsilon>0$ is the smallness parameter. We deduce from \Cref{S2T4} that the normal velocity in the mathematical approach is of order $\frac{1}{\omega}$. Therefore, if we consider the smallness parameter to be $\epsilon\leftrightarrow \frac{1}{\omega}$, then the physical structural assumption $|\vec{V}^{\perp}|\sim \epsilon v_t$ turns out to be a consequence in the mathematical approach. Consequently, imposing this condition on the structure of the solution is, a posteriori, not a restriction and hence justified.
	\item It was assumed that the gyrofrequency satisfies
	\begin{gather}
		\label{S4E10}
		\dot{\phi}\sim \Omega(\vec{R})\leftrightarrow \Omega(R_{\omega}(t))=\omega|B(R_{\omega}(t))|=\omega |B(x(t))|+o(\omega)\text{ as }\omega\rightarrow\infty
	\end{gather}
	where $\dot{\phi}$ denotes the angular velocity of the gyromotion. We note that in the mathematical calculations no angle has appeared explicitly. If we however let $e^2_{\omega}(t),e^3_{\omega}(t)$ be a family of orthonormal bases which span at each time $t$ the plane perpendicular to $b(R_{\omega}(t))$, then we can express
	\begin{gather}
		\nonumber
		\rho_{\omega}(t)=\frac{b(x_{\omega}(t))\times v_{\omega}(t)}{\omega|B(x_{\omega}(t))|}=\frac{b(R_{\omega}(t))\times v_{\omega}(t)}{|b(R_{\omega}(t))\times v_{\omega}(t)|}\frac{\sqrt{|v_0|^2-(b(R_{\omega}(t))\cdot v_{\omega}(t))^2}}{\omega|B(x_{\omega}(t))|}+o\left(\frac{1}{\omega}\right)
		\\
		\nonumber
		=\frac{b(R_{\omega}(t))\times v_{\omega}(t)}{|b(R_{\omega}(t))\times v_{\omega}(t)|}\frac{\sqrt{|v_0|^2-h^2(t)}}{\omega|B(x(t))|}+o\left(\frac{1}{\omega}\right)=u_{\omega}(t)\frac{\sqrt{|v_0|^2-h^2(t)}}{\omega|B(x(t))|}+o\left(\frac{1}{\omega}\right)\text{ as }\omega\rightarrow\infty
	\end{gather}
	where we set $u_{\omega}(t):=\frac{b(R_{\omega}(t))\times v_{\omega}(t)}{|b(R_{\omega}(t))\times v_{\omega}(t)|}$. We observe that $u_{\omega}(t)$ is perpendicular to $b(R_{\omega}(t))$ and of unit length so that we can express it as
	\begin{gather}
		\nonumber
		u_{\omega}(t)=\cos(\beta_{\omega}(t))e^2_{\omega}(t)+\sin(\beta_{\omega}(t))e^3_{\omega}(t)
	\end{gather}
	for suitable angles $\beta_{\omega}(t)$.
	
	We observe that the angles $\beta_{\omega}(t)$ depend on the chosen orthonormal frame $e^2_{\omega}(t),e^3_{\omega}(t)$. In particular, if we let $\alpha\equiv \alpha_{\omega}(t)$ be any other family of angles, then
	\begin{gather}
		\nonumber
		E^2_{\omega}(t):=\cos(\alpha)e^2_{\omega}(t)+\sin(\alpha)e^3_{\omega}(t)\text{, }E^3_{\omega}(t):=\sin(\alpha)e^2_{\omega}(t)-\cos(\alpha)e^3_{\omega}(t)
	\end{gather}
	forms another smooth family of orthonormal frames such that
	\begin{gather}
		\nonumber
		u_{\omega}(t)=\cos(\alpha-\beta)E^2_{\omega}+\sin(\alpha-\beta)E^3_{\omega}
	\end{gather}
	where $\beta\equiv \beta_{\omega}(t)$. Therefore, by letting $\alpha_{\omega}(t):=\beta_{\omega}(t)+\omega\int_0^t|B(x(s))|ds$, we obtain a family of orthonormal frames satisfying 
	\begin{gather}
		\label{Angle}
		u_{\omega}(t)=\cos(\phi_{\omega}(t))E^2_{\omega}(t)+\sin(\phi_{\omega}(t))E^3_{\omega}(t)\text{ where }\phi_{\omega}(t)=\omega\int_0^t|B(x(s))|ds.
	\end{gather}
	Consequently
	\begin{gather}
		\nonumber
		\dot{\phi}_{\omega}(t)=\omega|B(x(t))|=\Omega(R_{\omega}(t))+o(\omega)\text{ as }\omega\rightarrow\infty
	\end{gather}
	where we used the asymptotics from (\ref{S4E10}). We conclude that there exists a family of orthonormal frames $\{b(R_{\omega}(t)),E^2_{\omega}(t),E^3_{\omega}(t)\}$ such that the corresponding angles satisfy the relationship (\ref{S4E10}) as imposed in the physics literature approach. Thus, upon identifying $\phi_{\omega}\leftrightarrow \phi$, this assumption is also justified.
	
	As an example one can consider the most simple case $B(x)=(0,0,1)$ which can be solved analytically and for which one can explicitly compute
	\begin{gather}
		\nonumber
		|v^\perp_0|u_{\omega}(t)=\cos(\omega t)\begin{pmatrix}
			-v_{0,2} \\
			v_{0,1} \\
			0
		\end{pmatrix}+\sin(\omega t)\begin{pmatrix}
		v_{0,1} \\
		v_{0,2} \\
		0
		\end{pmatrix}
	\end{gather}
	where $v_0=(v_{0,1},v_{0,2},v_{0,3})$ is the initial velocity and $v^\perp_0=(v_{0,1},v_{0,2},0)$. We conclude that in this case
	\begin{gather}
		\nonumber
		E^2_{\omega}(t)=\begin{pmatrix}
			-v_{0,2} \\
			v_{0,1} \\
			0
		\end{pmatrix}\text{ and }E^3_{\omega}(t)=\begin{pmatrix}
		v_{0,1} \\
		v_{0,2} \\
		0
		\end{pmatrix}
	\end{gather}
	is a non-rotating coordinate system and so $\phi_{\omega}(t)=\omega\int_0^t|B(x(s))|ds$ describes a real physical rotation with respect to the above non-rotating frame.
	\item The last assumption which was made in the physics literature approach was
	\begin{gather}
		\label{S4E11}
		\dot{\rho}\sim \omega_B\rho
	\end{gather}
	where $\rho$ denotes the gyroradius and $\omega_B\sim \frac{v_t}{L_B}$. We recall that we had found the correspondence
	\begin{gather}
		\nonumber
		\rho\leftrightarrow \hat{\rho}_{\omega}(t)=\frac{\sqrt{|v_0|^2-h^2(t)}}{|B(x(t))|\omega}+o\left(\frac{1}{\omega}\right)\text{ as }\omega\rightarrow\infty.
	\end{gather}
	We notice that $\frac{|v_0|^2-h^2(t)}{|B(x(t)))|^2}=\frac{2\mu_0}{|B(x(t))|}$ where $\mu_0$ is the (time-independent) magnetic moment so that we arrive at
	\begin{gather}
		\label{S4E12}
		\hat{\rho}_{\omega}(t)=\frac{\sqrt{\frac{2\mu_0}{|B(x(t))|}}}{\omega}+o\left(\frac{1}{\omega}\right)\text{ as }\omega\rightarrow\infty.
	\end{gather}
	We observe further that $\frac{d}{dt}\frac{\hat{\rho}^2_{\omega}}{2}=\hat{\rho}_{\omega} \dot{\hat{\rho}}_{\omega}$ and so $\dot{\hat{\rho}}_{\omega}\sim \omega_B\hat{\rho}_{\omega}$ is equivalent to saying that $\frac{d}{dt}\hat{\rho}^2_{\omega}\sim \omega_B\hat{\rho}^2_{\omega}$ or again equivalently we can show that
	\begin{gather}
		\label{S4E13}
		\omega^2\frac{d}{dt}\hat{\rho}^2_{\omega}\sim \omega^2\omega_B\hat{\rho}^2_{\omega}.
	\end{gather}
	To this end we recall that according to (\ref{S4E12}) we have
	\begin{gather}
		\nonumber
		\omega^2\hat{\rho}^2_{\omega}(t)=2\mu_0|B(x(t))|^{-1}+o(1)\text{ as }\omega\rightarrow\infty.
	\end{gather}
	The leading order term is given by $2\mu_0|B(x(t))|^{-1}$ for which we compute
	\begin{gather}
		\nonumber
		2\mu_0 \frac{d}{dt}|B(x(t))|^{-1}=-\frac{2\mu_0}{|B(x(t))|^2}\dot{x}(t)\cdot \nabla |B|(x(t))=-\frac{2\mu_0 h(t)}{|B(x(t))|^2}b(x(t))\cdot \nabla |B|(x(t))
	\end{gather}
	where we used that $\dot{x}(t)=h(t)b(x(t))$. Consequently the leading order term
	\begin{gather}
		\nonumber
		\hat{\rho}^*_{\omega}(t):=\frac{\sqrt{\frac{2\mu_0}{|B(x(t))|}}}{\omega}
	\end{gather}
	in the expansion of $\hat{\rho}_{\omega}(t)$, cf. (\ref{S4E12}), satisfies
	\begin{gather}
		\label{S4E14}
		\frac{d}{dt}\hat{\rho}^*_{\omega}(t)=-\frac{h(t)}{2}\frac{b(x(t))\cdot \nabla |B|(x(t))}{|B(x(t))|}\rho^*_{\omega}(t).  
	\end{gather}
	We notice that according to (\ref{S4E14}) only the relative spatial rate of change in direction of the magnetic field is relevant for the temporal rate of change of the gyroradius so that one cannot expect that the rate of change of the (leading order) term $\rho^*_{\omega}(t)$ is proportional to
	\begin{gather}
		\nonumber
		\frac{1}{\frac{|B(x(t))|}{|\nabla |B|(x(t))|}}=\frac{1}{\frac{|B(R_{\omega}(t))|}{|\nabla |B|(R_{\omega}(t))|}}+o(1)=\frac{1}{L_B}+o(1)\text{ as }\omega\rightarrow\infty
	\end{gather}
	unless the parallel and orthogonal rate of change are of the same order of magnitude. However, we can always obtain the estimate
	\begin{gather}
		\label{S4E15}
		\left|\frac{d}{dt}\rho^{*}_{\omega}(t)\right|\leq \frac{|h(t)|}{2}\frac{\rho^{*}_{\omega}(t)}{\frac{|B(x(t))|}{|\nabla |B|(x(t))|}}.
	\end{gather} 
	 We recall that $\omega_B\sim \frac{v_t}{L_B}$ and that we had the correspondence $V_{\parallel}\leftrightarrow h$. Therefore, if one again imposes a Maxwellian distribution of the velocity within the plasma region one will find that $V_{\parallel}\sim v_t$ on average so that (\ref{S4E15}) implies that $\left|\frac{d}{dt}\rho^{*}_{\omega}(t)\right|\lesssim \frac{v_t}{L_B}\rho^*_{\omega}$ up to lower order terms. However, if the magnetic field variation in normal direction of the field is dominant, then one cannot improve this upper bound $\lesssim$ to a similarity $\sim$ in general. Thus, it follows from the mathematical approach that the physics literature assumption (\ref{S4E11}) is justified only for magnetic fields whose magnetic field strength variation is of the same order of magnitude as its parallel magnetic field strength variation. We emphasise that according to the mathematical approach, assumption (\ref{S4E11}) is not necessary to ensure the validity of the guiding centre approximation, cf. \Cref{S2T4}.
	 
	 We lastly notice once again that the condition $V_{\parallel}\sim v_t$ will be violated by trapped particles and that our mathematical approach does not require making any assumptions about the parallel speed $h(t)$ so that (\ref{S4E14}) and (\ref{S4E15}) are the "correct" version of condition (\ref{S4E11}) which is always valid under the sole assumption of a strong magnetic field. 
\end{enumerate}

\subsection{Final remarks}\label{s-final}
\subsubsection{Validity of the mathematical formulas}
The only assumption in the derivation of \Cref{S2T2} and \Cref{S2T4} that has been made was that $\omega=\frac{qB_{\operatorname{ref}}}{m}\gg 1$ where $B_{\operatorname{ref}}$ is some reference magnetic field strength. One may for instance take $B_{\operatorname{ref}}$ to be the average field strength in the region of interest. With this choice the strength of the normalised magnetic field $B$ will be $1$ on average and hence the terms $\frac{1}{|B(x(t))|}$ appearing in the expressions of (\ref{S2E1}), (\ref{S2E4}) should be roughly of order $\mathcal{O}(1)$ and hence solely $\frac{1}{\omega}$ controls the smallness of the terms. In particular, if $B_{\operatorname{ref}}$ is chosen to be the average of the magnetic field strength and the magnetic field $\bold{B}$ is multiplied by some constant $\alpha>0$, then the normalised fields $B$,$B_{\alpha}$ corresponding to $\bold{B}$ and $\alpha\bold{B}$ respectively coincide, while the reference gyrofrequency $\omega$ is multiplied by $\alpha$ and thus our expansions become more and more accurate.

In addition, we do not require to assume any kind of statistical distribution of the velocity and as a consequence our expansions from \Cref{S2T2} and \Cref{S2T4} remain valid in all physical situations in which an electric field is absent and a strong magnetic field is present.

\subsubsection{Gyromotion and particle path curvature vector}
We argue in this section that the direction of the gyromotion always occurs in the opposite direction of the curvature vector of the (full) particle trajectory $x_{\omega}(t)$ so that its direction is unaffected by the torsion of the particle trajectory. To see this we notice that $|\dot{x}_{\omega}(t)|=|v_{\omega}(t)|=|v_0|$ for all $t$ and $\omega$ by means of (\ref{S1E2}). Hence, the arc-length parametrisation $\tilde{x}_{\omega}(s)$ with arc-length parameter $s$ is given by
\begin{gather}
	\nonumber
	\tilde{x}_{\omega}(s)=x_{\omega}\left(\frac{s}{|v_0|}\right)
\end{gather}
and so we can compute the curvature vector according to
\begin{gather}
	\nonumber
	k(s)N(s)=\tilde{x}^{\prime\prime}_{\omega}(s)=\frac{1}{|v_0|^2}\ddot{x}_{\omega}\left(\frac{s}{|v_0|}\right)=\frac{\omega}{|v_0|^2}\dot{x}_{\omega}\left(\frac{s}{|v_0|}\right)\times B(\widetilde{x}_{\omega}(s)).
\end{gather}
Consequently, at position $x_{\omega}(t)$ which corresponds to the arc-length parameter $s=|v_0|t$ we find
\begin{gather}
	\nonumber
	k(x_{\omega}(t))N(x_{\omega}(t))=\frac{\omega}{|v_0|^2}v_{\omega}(t)\times B(x_{\omega}(t))
	\\
	\nonumber
	=-\frac{b(x_{\omega}(t))\times v_{\omega}(t)}{\omega |B(x_{\omega}(t))|}\frac{\omega^2}{|v_0|^2}|B(x_{\omega}(t))|^2=-\left(\frac{\omega |B(x_{\omega}(t))|}{|v_0|}\right)^2\rho_{\omega}(t)
\end{gather}
where $\rho_{\omega}(t)$ is the gyromotion, cf. \Cref{S2D3}. Since by definition the scalar curvature $k$ is non-negative we conclude that $\rho_{\omega}(t)$ always points in opposite direction of the curvature vector of $x_{\omega}(t)$.

\subsubsection{Curvature term and Lorentz force term}
It is customary in the physics literature to present the guiding centre expansion as in \Cref{S2T4} including a "curvature" term, $b\times \kappa$, and a grad-$B$ term, $b\times \nabla |B|$. However, in the context of plasma physics there is a more insightful way in which these terms may be expressed. We do this here for completion. To this end we recall that plasma equilibria are characterised by the balance law
\begin{gather}
	\label{S4E16}
	B\times \operatorname{curl}(B)=\nabla p
\end{gather} 
where $B$ is the magnetic field and $p$ is the pressure. In applications the pressure function $p$ foliates the plasma domain of interest into (compact) invariant tori so that motions along magnetic field lines will stay confined within some compact region. This is known as Arnold's structure theorem \cite[II Theorem 1.2]{AK21}. Similarly, since the pressure level sets are compact, any drift away from a magnetic field line, but which remains tangent to a pressure surface will stay confined in some finite region. In conclusion, the "bad" drift is the one orthogonal to the level sets of the pressure.

To describe this drift we use the vector calculus identity
\begin{gather}
	\nonumber
	\nabla \frac{|B|^2}{2}=\nabla_BB+B\times \operatorname{curl}(B).
\end{gather}
We may then write
\begin{gather}
	\label{S4E17}
	\frac{b(y)\times \nabla |B(y)|}{|B(y)|}=\frac{b(y)\times \nabla_BB(y)}{|B|^2}+\frac{b\times (B\times \operatorname{curl}(B))}{|B|^2}=b\times \kappa+\frac{b\times \nabla p}{|B|^2}
\end{gather}
according to (\ref{S4E16}) and where we used that $\kappa=\nabla_bb$ and so $\kappa=\nabla_bb=\frac{\nabla_bB}{|B|}+(b\cdot \nabla |B|^{-1})B$. We therefore arrive at the following equivalent expansion for the guiding centre motion
\begin{gather}
	\nonumber
	R_{\omega}(t)=x_0-\frac{b(x_0)\times v_0}{|B(x_0)|\omega}+\int_0^t(h(s)+o(1))b(R_{\omega}(s))ds
	\\
	\nonumber
	+\frac{1}{\omega}\int_0^t\left(\frac{|v_0|^2}{|B(R_{\omega}(s))|}-\mu_0\right)b(R_{\omega}(s))\times \kappa(R_{\omega}(s))ds
	\\
	\label{S4E18}
	+\frac{\mu_0}{\omega}\int_0^t\frac{b(R_{\omega}(s))\times (B(R_{\omega}(s))\times \operatorname{curl}(B)(R_{\omega}(s)))}{|B(R_{\omega}(s))|^2}ds+o\left(\frac{1}{\omega}\right)\text{ in }C^0_{\operatorname{loc}}(\mathbb{R},\mathbb{R}^3)\text{ as }\omega\rightarrow\infty
\end{gather}
where we used the definition of $\mu_0$ and that it is preserved in time. We recall that in the case of plasma equilibria $B\times \operatorname{curl}(B)=\nabla p$ is normal to the pressure level sets and hence the drift-term $b\times (B\times \operatorname{curl}(B))$ corresponds to a drift normal to the magnetic field lines but tangent to the pressure level sets.

Now if consider a regular pressure level set we find $\nabla p\neq 0$ on this level set and in particular $\left\{b,\operatorname{curl}^{\perp}(B),\nabla p\right\}$ with $\operatorname{curl}^{\perp}(B)=\operatorname{curl}(B)-\left(\operatorname{curl}(B)\cdot b\right)b$ forms an orthogonal basis at each point of our regular level set. We can hence expand $b(R_{\omega}(s))\times \kappa(R_{\omega}(s))$ in terms of this basis and see that the drift motion away from a given pressure level set is captured by the following expression
\begin{gather}
	\nonumber
	\frac{1}{\omega}\int_0^t\left(\frac{|v_0|^2}{|B(R_{\omega}(s))|}-\mu_0\right)\left((b(R_{\omega}(s))\times \kappa(R_{\omega}(s)))\cdot \nabla p (R_{\omega}(s))\right)\frac{\nabla p (R_{\omega}(s))}{|\nabla p(R_{\omega}(s))|^2}ds
	\\
	\label{S4E19}
	=\frac{1}{\omega}\int_0^t\left(\frac{|v_0|^2}{|B(x(s))|}-\mu_0\right)\left((B(x(s))\times \nabla p(x(s)))\cdot \nabla (|B|^{-1})(x(s))\right)\frac{\nabla p (R_{\omega}(s))}{|\nabla p (R_{\omega}(s))|^2}ds+o\left(\frac{1}{\omega}\right)
\end{gather}
where we used (\ref{S4E17}) in the last step and replaced the guiding centre by the zero order approximation. We notice now that the field lines $\gamma$ of the vector field $\frac{\nabla p}{|\nabla p|^2}$ satisfy $\frac{d}{dt}p(\gamma(t))=1$ and hence $p(\gamma(t))=t+p(\gamma(0))$ for all $t$. We conclude that $\frac{\nabla p}{|\nabla p|^2}$ describes the infinitesimal rate of change of the pressure. Consequently, its prefactor determines the full rate of change of the pressure along our guiding centre motion, i.e. the rate of change of the pressure is characterised by the quantity
\begin{gather}
	\label{S4E20}
	\left(\frac{|v_0|^2}{|B(x(s))|}-\mu_0\right)\left((B(x(s))\times \nabla p(x(s)))\cdot \nabla (|B|^{-1})(x(s))\right).
\end{gather}
Recently, \cite[Theorem 1.5]{BG25}, an expansion of the pressure along the full particle trajectory $x_{\omega}(t)$ has been derived in a rigorous manner by means of expanding $p(x_{\omega}(t))$ in terms of $\frac{1}{\omega}$. The result was the following
\begin{gather}
	\label{S4E21}
	p(x_{\omega}(t))-p(x_0)=\frac{(b(x(t))\times v_{\omega}(t))\cdot \nabla p(x(t))}{|B(x(t))|\omega}-\frac{(b(x_0)\times v_0)\cdot \nabla p(x_0)}{|B(x_0)|\omega}
	\\
	\label{S4E22}
	+\frac{1}{\omega}\int_0^t\left(\frac{|v_0|^2}{|B(x(s))|}-\mu_0\right)\left((B(x(s))\times \nabla p(x(s)))\cdot \nabla (|B|^{-1})(x(s))\right)ds+o\left(\frac{1}{\omega}\right).
\end{gather}
The first term on the right hand side of (\ref{S4E21}) is an oscillation and becomes negligible upon considering a time-average. The second term on the right hand side in (\ref{S4E21}) takes into account that the full particle trajectory and the guiding centre motion are initially displaced since the gyration has to be subtracted from the real particle position to obtain its guiding centre position. The actual rate of change of the pressure level set is therefore given by the integrand in (\ref{S4E22}) which coincides with the expression in (\ref{S4E20}). Hence, (\ref{S4E18}),(\ref{S4E19}) provide another way to determine the deviation from a given regular pressure surface. Most importantly, both approaches are consistent.

\section*{Acknowledgements}
This work has been partly supported by the  ANR-DFG project “CoRoMo”
ANR-22-CE92-0077-01, and has received financial support from the CNRS through the MITI interdisciplinary programs. This work has been supported by the Inria AEX StellaCage. The research was supported in part by the MIUR Excellence Department Project awarded to Dipartimento di Matematica, Università di Genova, CUP D33C23001110001.
\bibliographystyle{plain}
\bibliography{mybibfileNOHYPERLINK-u}

@article{alfven1940motion,
  author       = {Alfvén, Hannes},
  title        = {On the Motion of a Charged Particle in a Magnetic Field},
  journal      = {Arkiv för Matematik, Astronomi och Fysik},
  volume       = {27, Issue 22},
  year         = {1940},
  pages        = {1--20},
  publisher    = {Almqvist \& Wiksell},
  note         = {Early work introducing guiding-center concepts},
}

@article{northrop1961guiding,
  author       = {Northrop, T. G.},
  title        = {The Guiding Center Approximation to Charged Particle Motion},
  journal      = {Annals of Physics},
  volume       = {15},
  number       = {1},
  pages        = {79--101},
  year         = {1961},
  doi          = {10.1016/0003-4916(61)90167-1},
  note         = {Derives guiding-center velocity including drift terms},
}

@book{northrop1963adiabatic,
  author       = {Northrop, Theodore G.},
  title        = {The Adiabatic Motion of Charged Particles},
  publisher    = {Interscience Publishers},
  year         = {1963},
  series       = {Interscience Tracts on Physics and Astronomy},
  volume       = {21},
  address      = {New York},
  isbn         = {0598353933},
  note         = {Classic systematic derivation of guiding-center expansion and drift terms},
}

@article{littlejohn1983variational,
  author       = {Littlejohn, Robert G.},
  title        = {Variational Principles of Guiding Centre Motion},
  journal      = {Journal of Plasma Physics},
  volume       = {29},
  number       = {1},
  pages        = {111--125},
  year         = {1983},
  doi          = {10.1017/S002237780000060X},
  note         = {Hamiltonian/variational formulation of guiding-center theory},
}

@Book{Chen16,
  title     = {Introduction to {P}lasma {P}hysics and {C}ontrolled {F}usion},
subtitle = {2016},
  publisher = {Springer},
  year      = {2024},
  author    = {Chen, F. F.},
edition   = {Third},
}

@Book{AK21,
  title     = {Topological {M}ethods in {H}ydrodynamics},
subtitle = {},
  publisher = {Springer},
  year      = {2021},
  author    = {Arnold, V.I. and Khesin, B.A.},
edition   = {2nd},
}

@Article{BG25,
  author  = {Boscain, U. and Gerner, W.},
  title   = {Charged particle motion in a strong magnetic field: applications to plasma physics},
  journal = {Nonlinearity},
  year    = {2025},
  volume  = {38},
  number  = {10},
  pages   = {105017},
}

@ARTICLE{IGPW20,
       author = {Imbert-Gerard, L.-M. and Paul, E.J. and Wright, A.M.},
        title = "{An {I}ntroduction to {S}tellarators: {F}rom magnetic fields to symmetries and optimization}",
      journal = {arXiv e-prints},
         year = 2020,
        month = aug,
          eid = {arXiv:1908.05360},
        pages = {},
archivePrefix = {arXiv},
       eprint = {1908.05360},
 primaryClass = {plasm.ph},
}

@Book{IGPW24,
  title     = {An {I}ntroduction to {S}tellarators},
subtitle = {{F}rom magnetic fields to symmetries and optimization},
  publisher = {SIAM},
  year      = {2024},
  author    = {Imbert-Gerard, L.-M. and Paul, E.J. and Wright, A.M.},
edition   = {},
}
\footnotesize
\end{document}